\documentclass[12pt,psfig,times,mathptm]{article}

\usepackage{amssymb}
\usepackage{amsthm}
\usepackage{amsmath}
\usepackage{array}
\usepackage{url}
\usepackage{hyperref}
\hypersetup{breaklinks=true}
\usepackage{caption}
\usepackage{cancel}
\usepackage{subcaption}
\usepackage{color}
\usepackage{graphicx}
\usepackage{graphicx}
\usepackage{makecell}
\usepackage{float} 
\usepackage{framed}
\usepackage{multirow}
\usepackage[english]{babel}

\begin{document}
\baselineskip 0.6cm

\newcommand{\vertiii}[1]{{\left\vert\kern-0.25ex\left\vert\kern-0.25ex\left\vert #1 
    \right\vert\kern-0.25ex\right\vert\kern-0.25ex\right\vert}}
\newcommand{\beq}{\begin{equation}}
\newcommand{\eeq}{\end{equation}}
\newcommand{\bea}{\begin{eqnarray}}
\newcommand{\eea}{\end{eqnarray}}
\newcommand{\beas}{\begin{eqnarray*}}
\newcommand{\eeas}{\end{eqnarray*}}
\newcommand{\tri}{\triangleright}
\newcommand{\range}{{\rm range}}
\newcommand{\Ree}{{\rm Re }}
\newcommand{\Imm}{{\rm Im }}
\newcommand{\diag}{{\rm diag}}
\newcommand{\sign}{{\rm sign}}
\newcommand{\tr}{{\rm tr}}
\newcommand{\rank}{{\rm rank}}
\newcommand{\bp}{\bigskip}
\newcommand{\mdp}{\medskip}
\newcommand{\slp}{\smallskip}
\newcommand{\Rw}{\Rightarrow}
\newcommand{\ts}{& \hspace{-0.1in}}
\newcommand{\nn}{\nonumber}
\newtheorem{exa}{Example}[section]
\newtheorem{thm}{Theorem}[section]
\newtheorem{lem}{Lemma}[section]
\newtheorem{prop}{Proposition}[section]
\newtheorem{fact}{Fact}[section]
\newtheorem{cor}{Corollary}[section]
\newtheorem{defn}{Definition}[section]
\newtheorem{rem}{Remark}[section]

\def\QEDmark{\ensuremath{\nabla}}
\def\proof{\paragraph{Proof:}}
\def\endproof{\hfill\QEDmark}

\renewcommand{\theequation}{\thesection.\arabic{equation}}
\title{\textbf{A Discrete-Time Least-Squares Adaptive State Tracking
Control Scheme with A Mobile-Robot System Study
}}
\author{{\it Qianhong Zhao} and {\it Gang Tao}   \\
Department of Electrical and Computer Engineering\\
University of Virginia\\
Charlottesville, VA 22903, USA}
\date{}

\maketitle

\begin{abstract}
This paper develops an adaptive state tracking control scheme
for discrete-time systems, using the least-squares algorithm, as the new
solution to the long-standing discrete-time adaptive state tracking
control problem to which the Lyapunov method (well-developed for
the continuous-time adaptive state tracking problem) is not
applicable. The new adaptive state tracking scheme is based on a
recently-developed new discrete-time error model which has been used
for gradient algorithm based state tracking control schemes, and uses
the least-squares algorithm for parameter adaptation. The new
least-squares algorithm is derived to minimize an accumulative
estimation error, to ensure certain optimality for parameter
estimation. The system stability and output tracking properties are
studied. Technical results are presented in terms of plant-model matching,
error model, adaptive law, optimality formulation, and stability and
tracking analysis. The developed adaptive control scheme is applied
to a discrete-time multiple mobile robot system to meet an adaptive state
tracking objective. In addition, a collision avoidance mechanism is proposed to prevent collisions in the whole tracking process.    Simulation results are presented, which verify the
desired system state tracking properties under the developed
least-squares algorithm based adaptive control scheme.
\end{abstract}

\begin{quote}
\end{quote}

\section{Introduction}
\label{sec:introduction}

In the past decades, multi-mobile-robot path planning related researches attracted tremendous attention \cite{IMR1}, \cite{IMR5}, \cite{IMR4}, \cite{IMR3}, \cite{IMR2}, \cite{wetal20}, as its potential for future applications in the autonomous driving area. With the desired paths, the corresponding tracking schemes are needed to ensure that every robot moves as its corresponding desired path. 
Correspondingly, many multiple robot control related works have been done previously \cite{ftb03}, \cite{gcmc23}, \cite{ko02}, \cite{wg12}. Although past research demonstrated great results, additional solutions are necessary to be developed for the inevitable system uncertainties since sensor errors, robot internal structures, and, loads cause serious system uncertainties in practice. The adaptive tracking control technique is a proper solution for its ability to deal with system uncertainties.

Adaptive control has been a research hot spot for decades \cite{ht18}, \cite{KKK95}, \cite{na89}, \cite{st20}, \cite{t03}, \cite{zt22}. Adaptive state tracking control in \cite{mt98} and \cite{tjm01}, as one important adaptive control technique, is a candidate answer for the above multi-mobile-robot control problem. It offers a controller applying feedback structures and stable adaptive laws for the closed-loop system states to track the states from a selected reference model system with the existence of the system uncertainties. However, the existing adaptive state tracking control research usually focuses on continuous-time systems. In fact, in the robot control problems, it is more economical to implement discrete-time controllers because the computation burden will be considerably heavy when we decrease the sampling time to a small enough value to implement continuous-time controllers. Thus, the discrete-time adaptive state tracking controller is needed for such a multi-mobile-robot control problem. 

Some discrete-time tracking control related research results have been available now.  In \cite{DTST4intro1c01}, an adaptive discrete robust adaptive quasi-sliding-mode tracking control design is proposed for systems with unknown parameters, unmodeled dynamics, and bounded disturbances. In \cite{t03}, the adaptive state tracking control designs are introduced by an example of a single-input single-output (SISO) system case. In the example, gradient method based adaptive laws and state feedback control schemes, ensuring the tracking properties and the convergence of the parameter estimates, are demonstrated. In \cite{DTST4intro1c03}, an adaptive tracking control scheme based on reinforcement learning algorithms for discrete-time multi-input multi-output (MIMO) nonlinear discrete-time systems is presented, which needs two networks to generate the input and to monitor system performance respectively. To specifically ensure the desired
state tracking properties, \cite{T_DTASTC} starts the study of the discrete-time adaptive state tracking control problem by using gradient algorithms. 

This paper aims to propose a least-squares method based adaptive state tracking control algorithm for MIMO discrete-time linear time-invariant (LTI) systems, which can ensure the closed-loop system stability as well as the state tracking properties proved in \cite{T_DTASTC}, as one solution to the long-standing adaptive discrete-time state tracking control problem.  Differently, this paper focuses on using the least-squares algorithm to develop the adaptive law to estimate unknown parameters of the system for the desired tracking property by minimizing the cost function designed based on not only the incoming new data at every time step and all the historical data but also the previous parameter estimates. 

The following are the main contributions of this paper: \begin{itemize}
     \item Development of the adaptive law using the least-squares algorithm, based on the error model of the indirect adaptive control scheme for the MIMO systems with uncertainties, which can achieve both the state tracking property and the stability property.   

    \item Completion of the analysis for the state tracking performance of the adaptive control system, as well as the optimality and the stability of the adaptive law.

    \item Formulation of the specific solution to the multi-mobile-robot tracking control based on the proposed MIMO discrete-time adaptive state tracking control scheme and an additional mechanism to avoid possible collisions happening in the tracking process. 
\end{itemize}

\par The rest of this paper is organized as follows. The control problem to be solved in the paper is formulated in Section \ref{sec:TIsandBG}, where a review of the existing state tracking control schemes and the research motivation of this paper are given with corresponding related technical issues being listed. 
The new adaptive laws for the indirect adaptive control design of the discrete-time MIMO systems are derived in Section \ref{Sec:LSAL} by searching for the closed solution to minimize a quadratic cost function, which is the extension of the difference between the states of the controlled model and the reference (desired) model. Reasons why the new adaptive laws are designed for the indirect adaptive control design are also elaborated in Section \ref{Sec:LSAL}. Then, the corresponding optimality, stability and tracking properties are discussed. Section \ref{sec: RMandCA} presents the study of the robot model and proposes a collision avoidance mechanism. 
In Section \ref{sec: Simulation}, simulations on the multi-mobile-robot systems in \cite{zt23May} are demonstrated to verify the tracking performance of the proposed control algorithm.

\section{Problem Statement and Background}
\label{sec:TIsandBG}
This section formulates the long-standing adaptive discrete-time state tracking control problem and provides the research
background. In Section \ref{sec:PS}, the existing Lyapunov method based solution to the continuous-time adaptive state tacking problem and the control schemes using gradient algorithm based adaptive laws for the discrete-time case recently
developed in \cite{T_DTASTC} are reviewed in Sections \ref{Sec:continuous_solution} and \ref{sec:ETASTC} after the introduction of the adaptive state tracking
control problem to provide
some foundations for the least-squares algorithm adaptive control
schemes to be developed in this paper, whose motivation is discussed
in Section \ref{sec:MLSA}. Some related technical issues are discussed in Section
\ref{sec:TIs}, including the development of adaptive control schemes with certain
optimality and stability and the need of such schemes for a
multi-mobile-robot tracking control system which is modeled as a
discrete-time system.

\subsection{Discrete-Time Adaptive State Tracking Control Problem}
\label{sec:PS}
The continuous-time adaptive state tracking problem has been solved in
the literature by a Lyapunov-type adaptive control method which
however has not been successfully applied to discrete-time systems,
and the discrete-time adaptive state tracking control problem has been remained open. 

\medskip
{\textbf{Plant model description}}.
A discrete-time MIMO time-invariant plant is described as
\begin{equation}
    x(t+1) = Ax(t)+Bu(t),\, x(t) \in \mathbb{R}^n,\, u(t)\in \mathbb{R}^m,\, m>1, \,t =0,1,\ldots,
    \label{Plant}
\end{equation}
where $A\in \mathbb{R}^{n\times n }$ and $B\in \mathbb{R}^{n\times  m}$ are unknown constant parameter matrices, $x(t)$ is the plant state (output) vector, and $u(t)$ is the  input signal. For state feedback control, the plant state vector $x(t)$ is assumed to be available for measurement and is to be used for generating $u(t)$. 
The time variables $t$ and $t+1$ in the above discrete-time plant model \eqref{Plant} represent $kT$ and $ (k+1)T, \, k =0,1,\,\dots$, for simplification, with $T $ representing the sampling interval. 

\medskip
{\textbf{Control objective}}. The control objective is to design a state feedback control signal $u(t)$ to ensure that all closed-loop system signals are bounded and the system state vector $x(t)$ asymptotically tracks a reference state vector $x_m(t)$ generated from a chosen reference model system
\begin{equation}
\begin{split}
    x_m(t+1) = A_mx_m(t)+B_m r(t),\, x_m(t) \in \mathbb{R}^n,\, r(t)\in \mathbb{R}^m,
\end{split}\label{RefModel}
\end{equation}
where $A_m\in \mathbb{R}^{n\times n }$ and $B_m \in \mathbb{R}^{n\times  m}$ are some constant matrices, and $r(t)$ is a chosen reference input signal for desired system response. To design the control input $u(t)$ to meet the control objective, we need the following assumptions:
\begin{description}
\item[] {\bf Assumption (A1)}: All eigenvalues of $A_m$ are inside the unit circle of the complex plane;

\item[] {\bf Assumption (A2)}: The reference input signal $r(t)$ is bounded; 

\item[] {\bf Assumption (A3)}: There exist a constant matrix $K_1^*\in \mathbb{R}^{n \times m}$ and a non-singular constant matrix $K_2^*\in \mathbb{R}^{m\times m}$ such that
\begin{equation}
    A+BK^{*T}_1 = A_m,\; BK_2^* = B_m;
    \label{MatchCond_MIMOIND}
\end{equation}

\item[] {\bf Assumption (A4)}: In Assumption (A3), $K_2^* = \diag \{k_{21}^*, \ldots, k_{2m}^* \}$, and $\sign[k_{2i}^*],\;i = 1,\ldots,m$, are known.
\end{description}

Assumptions (A1)-(A2) ensure the stability of the reference system. Assumption (A3) guarantees the existence of the nominal control law achieving the control objective. Assumption (A4) is used for parameter adaption and adaptive control laws. 

\subsubsection{Basic Solutions for the Continuous-Time Case}
\label{Sec:continuous_solution}

The continuous-time version of the MIMO time-invariant plant \eqref{Plant} is 
\begin{equation}
    \Dot{x}(t) = Ax(t) + Bu(t),\, x(t) \in \mathbb{R}^n,\, u(t)\in \mathbb{R}^m,\, m\geq1,
    \label{Plant_cont}
\end{equation}
where $A\in \mathbb{R}^{n\times n }$ and $B\in \mathbb{R}^{n\times  m}$ are unknown constant parameter matrices, $x(t)$ is the plant state (output) vector, and $u(t)$ is the  input signal. For state feedback control, the plant state vector $x(t)$ is assumed to be available for measurement and is to be used for generating $u(t)$. Different from the discrete-time case, $t \geq 0$ is the time variable of continuous-time systems. 

The control objective for the state tracking control problem is to design $u(t)$ to ensure closed-loop system signal boundedness and asymptotic $x(t)$ tracking the state vector $x_m(t)
\in\mathbb{R}^n$ of a continuous-time reference model system 
\begin{equation}
\begin{split}
    \dot{x}_m(t) = A_mx_m(t)+B_m r(t),\, x_m(t) \in \mathbb{R}^n,\, r(t)\in \mathbb{R}^m,
\end{split}
\label{RefModel_cont}
\end{equation}
where $A_m\in \mathbb{R}^{n\times n }$ and $B_m \in \mathbb{R}^{n\times  m}$ are some constant matrices, and $r(t)$ is chosen reference input signal for desired system response.

Based on assumptions (A1)-(A4), the state feedback control law is 
\begin{equation}
    u(t) = K_1^{T}(t) x(t) + K_2(t)r(t),
    \label{AdaptiveU}
\end{equation}
where $K_{1}(t)\in \mathbb{R}^{n\times m}$ and $K_{2}(t) \in \mathbb{R}^{m\times m}$ are updated adaptively to estimate the nominal parameter matrices $K_{1}^*$ and $K^*_{2}$ respectively. 

The basic solution for the continuous-time state tracking control problem has been summarized in \cite{T_DTASTC} including the direct adaptive control design based on \cite{t03} and the indirect control design based on \cite{na89}.

\bigskip
\textbf{Direct adaptive control design}. For the state tracking error
\begin{equation}
    e(t) = x(t)-x_m(t)
    \label{et_cont}
\end{equation}
using \eqref{Plant_cont}-\eqref{AdaptiveU}, the tracking error dynamic is derived as 
\begin{equation}
    \Dot{e}(t) = A_mx(t) + B_m\left( K_2^{*-1}\Tilde{K}_1 (t)x(t)+K_2^{*-1}\Tilde{K}_2(t)r(t) \right),
    \label{edot_cont}
\end{equation}
where $\Tilde{K}_1(t) = K_1(t)-K_1^*$ and $\Tilde{K}_2(t) = K_2(t)-K_2^*$. 

The adaptive laws for the estimates $K_1(t)$ and $K_2(t)$ are selected as 
\begin{equation}
    \dot{K}_1^T(t) = -S^T_pB_m^TPe(t)x^T(t)
\end{equation}
\begin{equation}
    \dot{K}_2(t) = -S^T_pB_m^TPe(t)r^T(t),
\end{equation}
where $P= P^T>0$ satisfying $PA_m+A^TP = -Q$ for a chosen $Q=Q^T>0$, and $S_p = \diag\{\sign[k_{21}^*],\ldots,\sign[k_{2m}^*]\}$ such that
\begin{equation}
    M_s=K_2^*S_p = M^T_s>0.
\end{equation}

The time-derivative of the positive definite function 
\begin{equation}
    V = e^TPe+\tr[\Tilde{K}_1M_s^{-1}\Tilde{K}_1^T]+\tr[\Tilde{K}_2^TM_s^{-1}\Tilde{K}_2],
\end{equation}
can be derived as $\dot{V} = -e^T(t)Qe(t)\leq 0$, from which we have that $e(t)$, $\Tilde{K}_1(t)$ and $\Tilde{K}_2(t)$ are bound and $e(t)\in L^2$, that is, $x(t)$, $K_1(t)$ and $K_2(t)$ are bounded, and so is $u(t)$, that is, all closed-loop signals are bounded. From \eqref{edot_cont}, it follows that $\dot{e}(t)$ is bounded (so that $e(t)$ is uniformly continuous), and with $e(t)\in L^2$, we have $\lim_{t\rightarrow\infty}e(t) = 0$, according to Barbalat lemma \cite{t03}. 

\bigskip
\textbf{Indirect adaptive control design}.
With the Assumption (A3), the matching condition (\ref{MatchCond_MIMOIND}) is rewritten as 
\begin{equation}
    \begin{split}
        A = A_m - B_m\Theta_1^{*T},\;
        B  = B_m\Theta^*_2,
    \end{split}
\end{equation}
where 
\begin{equation}
\begin{split}
        \Theta_1^* = K_1^*\left( K_2^{*-1} \right)^T\in \mathbb{R}^{m\times n },\;
        \Theta_2^* = K_2^{*-1} \in \mathbb{R}^{m\times m} 
        . 
\end{split}
\end{equation}
Then, \eqref{Plant_cont} is parameterized as 
\begin{equation}
    \Dot{x}(t) = Ax(t)+Bu(t) = A_mx(t) + B_m(\Theta^*_2u(t) - \Theta_1^{*T}x(t) ). 
\end{equation}

Denoting $\Theta_1(t)$ and $\Theta_2(t)$ as the estimates of $\Theta_1^*$ and $\Theta_2^*$, based on \eqref{Model4IndirectDes_MIMO}, the state estimator generating an estimate $\hat{x}(t)$ of the plant state $x(t)$ can be designed as 
\begin{equation}
    \dot{\hat{x}} (t) = A_m\hat{x}(t) + B_m (\Theta_2(t) u(t) - \Theta_1^T(t) x(t)).
    \label{estimator_cont}
\end{equation}

For the state estimator error 
\begin{equation}
    e_x(t) = \hat{x}(t) - x(t),
    \label{EseDef_cont}
\end{equation}
the state estimator error equation is 
\begin{equation}
\begin{split}
    \dot{e}_x(t) = &\, A_m e_x(t) + B_m\big( (\Theta_2(t) - \Theta^*_2)u(t)-     (\Theta_1(t) - \Theta^*_1)^T    x(t)\big).
\end{split}
    \label{exA_cont}
\end{equation}

The adaptive laws for $\Theta_1(t)$ and $\Theta_2(t)$ are selected as 
\begin{equation}
    \Dot{\Theta}_1(t) = \Gamma_1 x(t)e_x^T(t)PB_m 
    \label{Theta1dot}
\end{equation}
\begin{equation}
    \Dot{\Theta}_2(t) = -\Gamma_2B_m^TPe_x(t)u^T(t)+F_2(t) ,
    \label{Theta2dot}
\end{equation}
where $\Gamma_1=\Gamma_1^T>0$, $\Gamma_2>0$ is diagonal, $P=P^T>0$ satisfying $PA_m+A^T_mP = -Q$ for a chosen $Q=Q^T>0$, and $F_2(t)$ is a projection signal to be designed. 

For the positive definite function 
\begin{equation}
    V = e_x^TPe_x + \tr[(\Theta_1-\Theta_1^{*})^T\Gamma_1(\Theta_1-\Theta_1^{*})] + \tr[(\Theta_2-\Theta_2^{*})^T\Gamma_2(\Theta_2-\Theta_2^{*})],
\end{equation}
we derive its time-derivative as $\dot{V}= -e_x^TQe_x\leq0$, from which we conclude that $\Theta_1(t)$, $\Theta_2(t)$ and $e_x(t)$ are all bounded, and that $e_x(t)\in L^2$. 

\bigskip
\textbf{Control law}. 
With $\Theta_1(t)$ and $\Theta_2(t)$ updated from the adaptive laws \eqref{Theta1dot}-\eqref{Theta2dot}, the control input is designed as 
\begin{equation}
    \begin{split}
        u(t) &= {\Theta_2^{-1}(t)} \left(\Theta_1^T(t)x(t) + r(t)\right). 
        \label{uID}
    \end{split}
\end{equation}  
To ensure the control law \eqref{uID} is meaningful, $\Theta_2(t)$ is guaranteed to be nonsingular for all $t\geq0$ by using a projection signal $F_2(t)
$ based on Assumption (A5) in addition to Assumption (A4).
\begin{description}
    \item [] \textbf{Assumption (A5)}: Upper bounds $k_{2i}^b$ of $|k_{2i}^*|$: $k_{2i}^b \geq |k_{2i}^*|$, $i=1,\ldots,m$, are known.
\end{description}

The projection signal $F(t)$ is chosen to be diagonal, $F_2(t) = \diag\{f_{21}(t),\ldots,f_{2m}(t)\}$, whose diagonal elements are set as 
\begin{equation}
    f_{2i}(t) = \left\{\begin{array}{ll}
         0&  \begin{array}{l}
              \text{if } \, \sign[\theta^*_{2i}]\theta_{2i}(t)>{1}/{k_{2i}^b}, \, \text{or}  \\
              \text{if } \, \sign[\theta^*_{2i}]\theta_{2i}(t)={1}/{k_{2i}^b}\; \,\text{and}\;\, \sign[\theta^*_{2i}]g_{2i}(t)\geq0
         \end{array}\\
         -g_{2i}(t) & \;\,\text{otherwise}, 
    \end{array}\right.
    \label{ProjectionSignal}
\end{equation}
where $g_{2i}$ denotes the diagonal elements of matrix $G_2(t) = \Gamma_2B_m^TPe_x(t)u^T(t)$ (with $\Gamma_2=\Gamma_2^T>0$ being diagonal) in \eqref{Theta2dot}. The projection signal constructed by \eqref{ProjectionSignal} guarantees that $\sign[\theta_{2i}(t)] = \sign[\theta^*_{2i}]$, $|\theta_{2i}(t)|\geq 1/k_{2i}^b>0$ and $(\theta_{2i}(t)-\theta_{2i}^*)f_{2i}(t)\leq0$. 

\medskip
With the control law \eqref{uID}, the estimator equation \eqref{estimator_cont} becomes 
\begin{equation}
    \dot{\hat{x}}(t) =A_m\hat{x}(t)+B_mr(t),
\end{equation}
that is, $\hat{x}(t)$ is bounded so that $x(t) = \hat{x}(t)-e_x(t)$, $u(t)$ and $\dot{x}(t)$ are bounded, and $\lim_{t\rightarrow\infty}(\hat{x}(t)-x_m(t)) = 0$ exponentially so that $\hat{x}(t)-x_m(t)\in L^2$. Hence we have that $x(t)-x_m(t)\in L^2$, and, with $\dot{x}(t)-\dot{x}_m\in L^{\infty}$, that $\lim_{t\rightarrow\infty}{x}(t)-x_m(t) = 0$.

Both the direct adaptive control design and the indirect adaptive control design for the continuous-time state tracking control problem are based on the Lyapunov method as the positive definite function $V$ is a Lyapunov function, which contains the full error signals and ensures $\dot{V}$ always not greater than zero, of an adaptive system. Although the well-developed Lyapunov method based adaptive control designs are widely used in continuous-time systems, they have not been applied to discrete-time state tracking control problems because there is not a certain choice of adaptive laws that can ensure function $V$ is nondecreasing and positive definite for discrete-time cases. 

\subsubsection{Gradient Algorithms for Discrete-time Systems}
\label{sec:ETASTC}
For the discrete-time state tracking control problems at the beginning of Section \ref{sec:PS}, \cite{T_DTASTC} offers the gradient algorithm based direct adaptive control design and indirect adaptive control design to be solutions.

\bigskip
 \textbf{{Direct Adaptive Control Design}}. For the discrete-time MIMO plant model \eqref{Plant}: 
 \begin{equation}
     x(t+1) = Ax(t)+Bu(t),\, x(t) \in \mathbb{R}^n,\, u(t)\in \mathbb{R}^m,\, m\geq1,
 \end{equation}
the reference system \eqref{RefModel}:
\begin{equation}
    x_m(t+1) = A_mx_m(t)+B_m r(t),\, x_m(t) \in \mathbb{R}^n,\, r(t)\in \mathbb{R}^m,
\end{equation}
 and the control law \eqref{AdaptiveU}:
 \begin{equation}
    u(t) = K_1^{T}(t) x(t) + K_2(t)r(t), t = 0,1,2,\ldots,
    \label{AdaptiveU_SIMO}
\end{equation}
where $K_{1}(t)\in \mathbb{R}^{n\times m}$ and $K_{2}(t) \in \mathbb{R}^{m\times m}$ are updated adaptively to estimate the nominal parameter matrices $K_{1}^*$ and $K^*_{2}$ satisfying \eqref{MatchCond_MIMOIND} in Assumption (A3). 

With the control law \eqref{AdaptiveU_SIMO} for the plant \eqref{Plant}, the tracking error $e(t) = x(t) - x_m(t)$ can be rewritten as 
\begin{equation}
    e(t+1) = A_m e(t) + B_m K_2^{*-1}\Tilde{\Theta}^T(t)\omega(t),
    \label{e(t+1)}
\end{equation}
where 
\begin{equation}
    \omega(t) = \left[x^T(t),r^T(t)\right]^T 
\end{equation} 
\begin{equation}
    \Tilde{\Theta}(t) = \Theta(t)-\Theta^*
\end{equation}
with 
\begin{equation}
    \Theta(t) = \left[K_1^T(t),K_2(t)\right]^T,\;\Theta^*=\left[K_1^{*T},K_2^*\right]^T. 
\end{equation}
Denoting $\rho_i^* = 1/k^*_{2i},\;i = 1,\ldots,m$, and 
\begin{equation}
    \Theta(t) = \left[\theta_1(t),\ldots,\theta_m(t)\right],\;\Theta^* = \left[\theta_1^*,\ldots,\theta_m^*\right],
\end{equation}
with $W_m(z) = (zI-A_m)^{-1}B_m$, $\Tilde{\theta}_i(t) = \theta_i(t)-\theta_i^*,\,i=1,\ldots,m$, and $K_2^{-1} = \diag\left\{\rho_1^*,\ldots,\rho_m^*\right\}$, we express \eqref{e(t+1)} as 
\begin{equation}
    e(t) = W_m(z)\begin{bmatrix}
        \rho_1^*\Tilde{\theta}_1^T\omega\\
        \vdots\\
        \rho_m^*\Tilde{\theta}_m^T\omega\\
    \end{bmatrix}
   (t),
   \label{etW}
\end{equation}
which, in terms of $e(t) = \left[e_1(t),\ldots,e_n(t)\right]^T$ and $W_m(z) = [w_{ij}(z)],\,i=1,\ldots,n,\,j=1,\ldots,m$, can be further written as 
\begin{equation}
    \begin{split}
        e_i(t) = \sum_{j=1}^m\rho_j^*w_{ij}(z)\left[\Tilde{\theta}_j^T\omega\right](t).
        \label{te_para}
    \end{split}
\end{equation}

Introducing the auxiliary signals  
\begin{equation}
    \zeta_{ij}(t) = w_{ij}(z)[\omega](t) \in \mathbb{R}^{n+m}
\end{equation}
\begin{equation}
    \xi_{ij}(t) = \theta_j^T(t)\zeta_{ij}(t)-w_{ij}(z)\left[ \theta_j^T\omega\right](t)\in\mathbb{R},
\end{equation}
we define
the estimation error $\epsilon(t) = [\epsilon_1(t),\ldots,\epsilon_n(t)]$ as
\begin{equation}
    \begin{split}
        \epsilon_i(t) & = e_i(t)+ \sum_{j=1}^m\rho_j(t)\xi_{ij}(t),
    \end{split}
    \label{eps_def_DD}
\end{equation}
where $\rho_j(t), \,j = 1,\ldots,m,$ are the estimates of $\rho^*_j$.
With \eqref{te_para}, we have
\begin{equation}
  \epsilon_i(t)  = \sum_{j = 1}^m\rho_j^*\Tilde{\theta}_j(t)^T\zeta_{ij}(t) + \sum_{j = 1}^m(\rho_j(t)-\rho_j^*)\xi_{ij}(t).
        \label{parametered_eps_DD}
\end{equation}

\medskip

The discrete-time gradient-type adaptive laws for $\theta(t)$ and $\rho(t) $, which minimize the cost function $J = \frac{1}{2} \frac{\sum_{i=1}^n\epsilon_i^2}{m^2}$, are designed as
\begin{equation}
    \theta_i(t+1) = \theta_i(t)-\frac{\sign[\rho^*_i]\Gamma_i\sum_{k=1}^n\epsilon_k(t)\zeta_{ki}(t)}{m^2(t)}
    \label{thetaupdate_G}
\end{equation}
\begin{equation}
    \rho_i(t+1) = \rho_i(t) -\frac{\gamma_i \sum_{k=1}^n\epsilon_k(t)\xi_{ki}(t)}{m^2(t)},
\end{equation}
where $0<\Gamma_i = \Gamma_i^T|\rho^*_i|<2 I$, $0<\gamma_i<2$, $i = 1,\ldots,m$, and
\begin{equation}
    m(t) = \sqrt{1+ \sum_{i=1}^n\sum_{j=1}^m\left(\zeta_{ij}^T(t)\zeta_{ij}(t) + \xi_{ij}^2(t)\right)}.
\end{equation}

\medskip
To implement the adaptive law \eqref{thetaupdate_G}, we need the following assumption in addition to Assumption (A4):  
\begin{description}
    \item[ ] \textbf{Assumption (A6)} Lower bounds $k^a_{2i}>0$ of $|k^*_{2i}|$: $|k_{2i}^*|\geq k^a_{2i},\; i = 1,\ldots,m$, are known. 
\end{description}

With Assumption (A6), the value of the gain matrix $\Gamma_i$, $0<\Gamma_i = \Gamma_i^T|\rho_i^*|<2 I$, can be specified as $0<\Gamma_i = \Gamma_i^T<k_{2i}^a I$, given that $\rho_i^* = \frac{1}{k_{2i}^*}$.

\bigskip
\textbf{Indirect adaptive control design}. 
With the Assumption (M3), the matching condition (\ref{MatchCond_MIMOIND}) is rewritten as 
\begin{equation}
    \begin{split}
        A = A_m - B_m\Theta_1^{*T},\;
        B  = B_m\Theta^*_2,
    \end{split}
    \label{NewMatCond_MIMO}
\end{equation}
where 
\begin{equation}
\begin{split}
        \Theta_1^* = K_1^{*}(K_2^{*-1})^T \in \mathbb{R}^{ n \times m},\;
        \Theta_2^* = K_2^{*-1} \in \mathbb{R}^{m\times m} 
        . 
\end{split}
\end{equation}
Equation \eqref{NewMatCond_MIMO} expresses the parameter uncertainties of $A$ and $B$ in the plant \eqref{Plant} in terms of the uncertainties of $\Theta_1^*$ and $\Theta_2^*$. Thus, the parametrized plant is 
\begin{equation}
    x(t+1) = Ax(t)+Bu(t) = A_m x(t) + B_m(\Theta_2^*u(t) - \Theta_1^{*T}x(t)),
    \label{Model4IndirectDes_MIMO}
\end{equation}
which can be used to design an adaptive parameter estimation scheme.

Denoting $\Theta_1(t)$ and $\Theta_2(t)$ as the estimates of $\Theta_1^*$ and $\Theta_2^*$, based on \eqref{Model4IndirectDes_MIMO}, the state estimator generating an estimate $\hat{x}(t)$ of the plant state $x(t)$ can be designed as 
\begin{equation}
    \hat{x} (t+1) = A_m\hat{x}(t) + B_m (\Theta_2(t) u(t) - \Theta_1^T(t) x(t)).
    \label{estimatorequatin}
\end{equation}
With $\Theta_1(t)$ and $\Theta_2(t)$ generated from adaptive laws to be specified, the control input is designed as 
\begin{equation}
    \begin{split}
        u(t) &= {\Theta_2^{-1}(t)} \left(\Theta_1^T(t)x(t) + r(t)\right). 
        \label{uID_new}
    \end{split}
\end{equation}  

For the state estimator error 
\begin{equation}
    e_x(t) = \hat{x}(t) - x(t),
    \label{EseDef}
\end{equation}
the state estimator error equation is 
\begin{equation}
\begin{split}
    e_x(t+1) = &\, A_m e_x(t) + B_m\big( (\Theta_2(t) - \Theta^*_2)u(t)-     (\Theta_1(t) - \Theta^*_1)^T    x(t)\big).
\end{split}
    \label{exA}
\end{equation}

Introducing the parameter matrices 
\begin{equation}
\begin{split}
    \Theta(t) &= [\Theta_1^T(t), \Theta_2(t)]^T= [\theta_1(t),
    \ldots, \theta_{m}(t)] \in \mathbb{R}^{(n+m ) \times m },
\end{split}
\label{ThetaDefIn}
\end{equation}
\begin{equation}
\begin{split}
    \Theta^* &= [\Theta_1^{*T}, \Theta_2^*]^T = [\theta^*_1,
    \ldots, \theta_{m}^*] \in \mathbb{R}^{(n+m ) \times m },
\end{split}
\label{Theta_star_intro}
\end{equation}
\begin{equation}
    \begin{split}
    \Tilde{\Theta}(t) &= \Theta(t) - \Theta^*= [\Tilde{\theta}_1(t),
    \ldots, \Tilde{\theta}_{m}(t)] \in \mathbb{R}^{(n+m ) \times m },
\end{split}
\end{equation}
and the vector signal 
\begin{equation}
    \omega(t) = [-x^T(t), u^T(t)]^T\in \mathbb{R}^{n+m},
\end{equation}
 we can rewrite \eqref{exA} as 
\begin{equation}
    e_x(t+1) =  A_m e_x(t) + B_m\big( \Tilde{\Theta}^T(t)\omega(t)\big). 
\end{equation}

With the $n \times m$ transfer matrix 
\begin{equation} 
\begin{split}
    W_m(z) & = (zI-A_m)^{-1}B_m,
\end{split}
\label{WmDef}
\end{equation}
we obtain 
\begin{equation}
    e_x(t) = W_m(z)\begin{bmatrix}
        \Tilde{\theta}_1^T\omega\\
        \vdots\\
        \Tilde{\theta}_m^T\omega
    \end{bmatrix}(t),
\end{equation}
which, with $e_x( t) = [e_{x1}(t),
\ldots,e_{xn}(t) ]\in \mathbb{R}^n$, is equivalent to 
\begin{equation}
\begin{split}
    e_{xi}(t) =& \sum_{j=1}^m w_{ij}(z)\left[\Tilde{\theta}_j^T\omega\right](t) ,
\end{split}
    \label{exEleID_MIMO}
\end{equation}
where $w_{ij}(z), i = 1,\ldots,n, j = 1, \ldots,m,$ is the $ij$th component of the transfer matrix $W_m(z)$. 

With the auxiliary signals 
\begin{equation}
    \xi_{ij}(t)= \theta_j^T(t)\zeta_{ij}(t) - w_{ij}(z) [\theta_j^T\omega](t),
    \label{XiID_MIMO}
\end{equation}
\begin{equation}
    \zeta_{ij}(t) = w_{ij} (z)[\omega](t) ,
    \label{ZetaID_MIMO}
\end{equation}
the estimation error is defined as 
\begin{equation}
    \epsilon_i(t) = e_{xi}(t) + \sum_{j=1}^m \xi_{ij}(t).
    \label{EstErrorID_MIMO}
\end{equation}

Substituting \eqref{exEleID_MIMO}-\eqref{ZetaID_MIMO} to \eqref{EstErrorID_MIMO}, we have 
\begin{equation}
\begin{split}
        \epsilon_i (t) &=\sum_{j=1}^m \Tilde{\theta}_{j}^T(t)\zeta_{ij}(t). \\
\end{split}
\label{parametered_eps_ID}
\end{equation}

\medskip
The adaptive laws are chosen as 
\begin{equation}
    \theta_i(t+1) = \theta_i(t)-\frac{\Gamma_i\sum_{k=1}^n\epsilon_k(t)\zeta_{ki}(t)}{m^2(t)},
    \label{ThetaUpdateG}
\end{equation}
where $0<\Gamma_i = \Gamma_i^T<2I$, $i = 1,\ldots,m$, and
\begin{equation}
    m(t) = \sqrt{1+ \sum_{i=1}^n\sum_{j=1}^m\left(\zeta_{ij}^T(t)\zeta_{ij}(t) + \xi_{ij}^2(t)\right)}.
\end{equation}

\medskip
It is crucial that the parameter matrix $\Theta_2(t)$ is nonsingular,
in order to implement the indirect adaptive control law \eqref{uID_new} with
$\Theta_2^{-1}(t)$ bounded. This can be achieved by parameter projection similar to the indirect adaptive control design for continuous-time systems introduced in Section \ref{Sec:continuous_solution},
using Assumption (A4) that $\Theta_2^* = K_2^*$ is diagonal to set
$\Theta_2(t)$ to be diagonal, using $\Gamma_i = \diag\{\Gamma_{i1},
\Gamma_{i2}\}$ with $\Gamma_{i2} \in R^{m \times m}$ being diagonal,
and using parameter projection on the diagonal components
$\Theta_{2i}(t)$ of $\Theta_2(t)$ to make $\sign[\Theta_{2i}(t)] =
\sign[k_{2i}^*]$ and $|\Theta_{2i}(t)| \geq
1/k_{2i}^b > 0$ (see Assumption (A5)). The detailed projection signal, that ensures the diagonal components of $\Theta_2(t)$ away from $0$, is similar to the signal $F_2(t)$ in \eqref{Theta2dot} whose components are designed as in \eqref{ProjectionSignal}.

\medskip
The parameter projection and stability analysis of both the direct adaptive control design and the indirect adaptive control design are detailed in \cite{T_DTASTC}.

\medskip
\subsubsection{\textbf{Motivation of the Least-Squares Algorithm}}
\label{sec:MLSA}

Besides the gradient algorithm based adaptive laws, there is another algorithm to derive the adaptive laws that ensure the system stability and achieve tracking objectives.

According to \cite{bv04}, the least-squares algorithm is a long-established and widely used estimation algorithm that can minimize the sum of the squared cost function. 
A comprehensive study of least-squares algorithms for
parameter estimation is given in \cite{gs84}. The classical least-squares algorithm is computationally efficient and easy to implement since it can be solved analytically using a closed-form solution that does not require iterative optimization.

In \cite{t03} and \cite{tz23May}, some modifications to the classical least-squares algorithm are done to make use of the incoming new data and previous parameter estimates iteratively by selecting a superposition of errors at different time instants as the cost function to update the estimates of unknown parameters. 
For example, the cost function for the estimation error $\epsilon(t) = \theta^T(t) \zeta(t) - y(t)$, where $\theta(t)$, $\zeta(t)$ and $y(t) = \theta^{*T}\zeta(t)$ are the unknown parameter vector, the measured vector signal, and the measured output respectively, is selected as $J = \frac{1}{2}\sum_{\tau=1}^t \left(\theta^T(t) \zeta(\tau) - y(\tau)\right)^2   $ rather than the simple quadratic form of estimation error $J = \frac{1}{2}\epsilon^2(t)$.

The modified algorithm becomes a batch-data algorithm that can update the parameter estimates by minimizing the iterative cost function, which is different from the parameter adaptive laws that apply the idea to force the estimates to update along the steep descent direction of the cost function (quadratic errors at last time instant) in \cite{T_DTASTC}. The modified algorithm is applicable to update the estimates of the unknown parameters in adaptive control problems. 

After analyzing the existing research about the adaptive state tracking problems, we find the use of least-squares algorithms for adaptive control is less often seen in the literature than the use of gradient algorithms.
This motivates our new research on developing adaptive state tracking control schemes using least-squares algorithms, with application for the multi-mobile-robot control problems, which can also ensure the closed-loop system stability and the desired tracking performance.

\setcounter{equation}{0}
\subsection{Technical Issues}
\label{sec:TIs}
A new gradient-type adaptive control scheme is recently developed in
\cite{T_DTASTC} to solve the discrete-time adaptive state tracking control
problem. A new least-squares type adaptive control scheme is developed
in this paper to also solve the discrete-time adaptive state tracking
control problem.
Specifically, we solve two new technical issues in this paper for discrete-time adaptive state tracking control. The details are illustrated as follows. 

\medskip
\textbf{\textit{Issue I: Discrete-time adaptive controllers using the least-squares algorithm}}. In real-life scenarios, we usually need to care about all the system states rather than the system outputs only. For example, it is necessary to ensure both the robot's position and velocity, which are the states of a robot system, to track desired positions and velocity trajectories if we want the robot to move in the desired path. For the adaptive state tracking control problems introduced at the beginning of this section, our control objective is to design an adaptive control $u(t)$ such that all the states of the unknown plant $x(t+1)  = Ax(t) + Bu(t)$ can track states of the known reference model $x_m(t+1) = A_mx_m(t) + B_m r(t)$, i.e. $\lim_{t\rightarrow\infty} \left(x(t)-x_m(t)\right) =0$, and ensure closed-loop stability. 

We choose the least-squares algorithm to design the unknown parameter adaptive law for the control input $u(t)$ as its potential to minimize the accumulated estimation error. Since the least-squares algorithm is different from the gradient algorithm, the introduced estimation errors (\eqref{parametered_eps_DD} or \eqref{parametered_eps_ID}) of two adaptive control designs need to be reconsidered to find out which adaptive control design is possible to design a least-squares algorithm based adaptive laws for the unknown parameters $\Theta^*$. The corresponding iterative adaptive laws should be derived by making the updated parameters guarantee the cost function, which is extended from accumulated estimation errors, is minimized at every iteration. Lastly, the optimality of the proposed adaptive laws and the overall tracking performance need to be proved theoretically.

\medskip
\textbf{\textit{Issue II: Discrete-time adaptive controller applied on a multi-mobile-robot system}}. For multi-mobile-robot control problems, we expect robots to move as desired paths without any collisions. With the developed adaptive state tracking control scheme using the least-squares algorithm for MIMO systems, the multi-mobile-robot control problems with parameter uncertainties should solved partially when the robot position and velocity are considered as system states and the reference model states are used to depict positions and velocities of the desired paths.

However, the individual adaptive state tracking control cannot guarantee that collision will not happen even though the reference models for different robots are selected collision-free as the adaptive control scheme needs some time to achieve the tracking objective. To eradicate the possible collisions in this process, an additional collision avoidance mechanism is necessary to develop.  

A simulation study on a 3-robot system in \cite{zt23May}, which is simplified from a multi-vehicle model of autonomous driving problems, is conducted to verify the performance of the proposed control structures and adaptive laws. The adaptive state tracking control problem is a basic part of the autonomous driving problems that ensure the controlled vehicles move as the designed trajectories.  

The simulation results also prove that the combination of the proposed adaptive control scheme and the collision mechanism can achieve the control objectives.

\setcounter{equation}{0}
\section{Least-Squares Adaptive Control Algorithm}
\label{Sec:LSAL}
{This section offers the solution to the first technical issue: an adaptive control scheme based on the least-squares algorithm is proposed for the long-standing discrete-time adaptive state tracking control problem based on the estimation error derived in \cite{T_DTASTC}. In Section \ref{Sec:ACDS}, the reason why indirect adaptive control design is suitable to implement the least-squares algorithm based adaptive laws is discussed. In Section \ref{Sec: PEEE}, the estimation error $\epsilon(t)$ introduced in Section \ref{sec:ETASTC}, is parameterized further for the development of the adaptive laws using the least-squares algorithm. The cost function selection, the new adaptive laws themselves, the adaptive law optimality analysis and the overall system tracking performance evaluation are shown in Section \ref{Sec:ALDA}.

 }

\subsection{Cost Function Selection for Adaptive Law Design}
\label{Sec:ACDS}
It is necessary to study the difference between the least-squares algorithm and the gradient algorithm to identify whether there are some additional requirements while implementing the least-squares algorithm to adapt the parameter estimates in the state tracking control design.

\medskip
\textbf{Cost function for the gradient algorithm.} 
Recalling from \eqref{parametered_eps_ID}, we have the estimation error expressions
\begin{equation}
        \epsilon_i(t) = \big(\theta_{1}(t)-\theta_{1}^*\big)^T \zeta_{i1}(t)   + \cdots +  \big(\theta_{l}(t)-\theta_{l}^*\big)^T \zeta_{il}(t), \, i = 1,\ldots,n,
        \label{estimationerrorgradient}
    \end{equation}
which can be written in the compact form
\begin{equation}
\begin{split}
    \epsilon_i(t) = \big(\theta(t)-\theta^*\big)^T \zeta_i(t) = \theta(t)\zeta_i(t) - y_i(t), \, i = 1,\ldots,n,
\end{split}
    \label{estimationerrorgeneral}
\end{equation} 
with $\theta(t)\in \mathbb{R}^{n_{\theta}}$ being the estimate of the unknown parameter $\theta^*\in \mathbb{R}^{n_{\theta}}$. In \eqref{estimationerrorgeneral}, vector signals $y_i(t) = \theta^*\zeta_i(t) $ and $\zeta_i(t)$ are real-time updated measurable signals. 

With $\epsilon(t) = [\epsilon_1(t),  \epsilon_2(t), \ldots,
\epsilon_n(t)]^T$, a normalized quadratic cost function of the parameter estimate $\theta(t)  $ is selected as
\begin{equation}
    J(\theta) = \frac{\epsilon^T(t)\epsilon(t)}{2m^2(t)},
    \label{CostGradient}
\end{equation}
where $m(t)$ is a so-called normalizing signal that does not depend on the parameter estimates $\theta(t)$ and ensures the boundedness of $\frac{\zeta_{i}\zeta^T_{i}}{m^2},\, i = 1,\ldots,n$. Then, the gradient-type iterative adaptive laws for the parameter estimates $\theta(t)$ are deduced by updating them in the steepest descent direction 
    $-\frac{\partial J}{\partial \theta}$ of $J(\theta)$
to minimize the cost function iteratively.

\medskip
\textbf{Cost function and optimality for the least-squares algorithm.}
To minimize the same estimation error, which has the form of \eqref{estimationerrorgeneral}, the least-square algorithm needs to set its cost function to be the normalized squared of a modified version of the accumulated estimation error
with a penalty of the initial estimate $\theta_0$:
\begin{equation}
    J(\theta) =\frac{1}{2} \sum_{\tau=0}^{t-1} \frac{1}{\kappa}{\sum_{i=1}^n\left(\theta^T(t) \zeta_i(\tau) - y_i(\tau)\right)^2 } + \frac{1}{2}\big(\theta(t)-\theta_0\big)^T P_0\big(\theta(t)-\theta_0\big), \kappa> 0 , P_0=P_0^T>0,
    \label{Cost_GLS}
\end{equation}
where the term $\frac{1}{2}\sum_{\tau=0}^{t-1} \frac{1}{\kappa}{\sum_{i=1}^n\left(\theta^T(t) \zeta_i(\tau) - y_i(\tau)\right)^2 } $ is a modified version of the accumulated estimation error norm over a signal measurement interval $\tau = 0, 1,\ldots ,t-1$:
\begin{equation}
   \sum_{\tau=0}^{t-1} \sum_{i=1}^n \left( \theta^T(\tau) \zeta_i(\tau) - y_i(\tau) \right)^2 (\tau) =  \sum_{\tau=0}^{t-1} \sum_{i=1}^n \epsilon_i^2 (\tau) =  \sum_{\tau=0}^{t-1} \epsilon^T (\tau) \epsilon(\tau), 
\end{equation}
with $\zeta_i(\tau)$ and $y_i(\tau)$ being the measured signals.

The optimal parameter $\theta(t)$, minimizing the cost function $J(\theta)$, is found by solving the equation 
\begin{equation}
    \frac{\partial J}{\partial \theta} = 0.
\end{equation}

After examining the estimation errors of the direct adaptive control design and the indirect adaptive control design in \cite{T_DTASTC}, we see that the estimation error in the indirect adaptive control design is suitable for developing the least-squares based adaptive laws for estimating the unknown parameters. The reason is that the estimation error $\epsilon(t)$ of the indirect adaptive control design has the same form as in \eqref{estimationerrorgeneral}, but the estimation error $\epsilon(t)$ of direct adaptive control does not have that form as the existence of $\rho_j(t)$ and $\rho^*_j$ $, j = 1,\ldots,m$. Moreover, the iterative solution such that the derivative of the cost function equals 0 will always contain unknown part $\rho_j^*$ while subsisting the estimation error of the direct adaptive design \eqref{parametered_eps_DD} into the cost function \eqref{Cost_GLS}.

\subsection{Parameterized Estimation Error Equation}
\label{Sec: PEEE}
In this subsection, the parameterized estimation error for systems with diagonal $K_2^*$ matrices of indirect adaptive control design are studied, which are crucial for developing the least-squares algorithm based adaptive laws for the control problems detailed in Section \ref{sec:PS} with Assumptions (A1)-(A5).

\bigskip
\textbf{Parameterized estimation error vector for systems with the diagonal $K_2^*$}. According to Assumption (M4), $\Theta_2(t)$ should be diagonal. Thus, we only need to estimate the diagonal elements of $\Theta_2(t)$ to reduce the computation burden. For this case, we present a special reduced-order adaptive controller for systems with the unknown diagonal $K_2^*$. 

 At first, we set $\Theta_1(t)$ and $\Theta_2(t) = \diag\{ \theta_{21}(t), \ldots \theta_{2m}(t) \}$ to be the estimates of the unknown parameter matrices $\Theta_1^*$ and $\Theta_2^*$. Based on \eqref{exA} and \eqref{WmDef}, the state estimation error for the system with the diagonal $K_2^*$ has the form of 
\begin{equation}
    e_x (t) = W_m(z)[(\Theta-\Theta^*)^T\omega](t) = W_m(z) \begin{bmatrix}
        \big({{\theta}}_1(t) - {\theta}_1^* \big)^T {\omega}_1(t)\\
        \vdots\\
        \big({{\theta}}_m(t) - {\theta}_m^* \big)^T {\omega}_m(t)
    \end{bmatrix} ,
\end{equation}
where vector ${\theta}_i(t)$ is the $i$th column of  
\begin{equation}
    {\Theta}(t) = \left[\Theta_1^T(t), \begin{bmatrix}
        \theta_{21}(t)\\
        \vdots\\
        \theta_{2m}(t)
    \end{bmatrix}\right]^T = [ {\theta}_1(t),\ldots,{\theta}_m(t)] \in \mathbb{R}^{(n+1) \times m},
    \label{BThetaDef}
\end{equation}
vector ${\theta}_i^*$ is the $i$th column of  
\begin{equation}
    {\Theta}^* = \left[\Theta_1^{*T}, \begin{bmatrix}
        \theta_{21}^*\\
        \vdots\\
        \theta_{2m}^*
    \end{bmatrix}\right]^T = [ {\theta}_1^*,\ldots,{\theta}_m^*]\in \mathbb{R}^{(n+1) \times m},
\end{equation}
and the signal ${\omega_i}(t)$ is defined as 
\begin{equation}
    {\omega_i}(t) = [-x^T(t),u_i(t)]^T\in \mathbb{R}^{n+1 }, i = 1,\ldots,m,
    \label{Def_Omegai}
\end{equation}
where $u_i(t)$ denotes the $i$th component of the input signal $u(t)$.

For the element of the state estimation error from $e_x( t) = [e_{x1}(t),\ldots,e_{xn}(t) ]\in \mathbb{R}^n$, we have 
\begin{equation}
    e_{xi}(t) = \sum_{j=1}^m  w_{ij}(z)\left[({\theta}_j-{\theta}^*_j)^T{\omega}_j\right](t), 
    \label{exEleID_MIMO_Diag}
\end{equation}
where $w_{ij}, i = 1,\ldots,n, j = 1, \ldots,m,$ is the $ij$th component of the transfer matrix $W_m(z)$.

With ${\theta}_j(t)$, ${\theta}^*_j$ and ${\omega}(t)$, the auxiliary signals \eqref{XiID_MIMO} and \eqref{ZetaID_MIMO} can be rewritten as 
\begin{equation}
    {\xi}_{ij}(t) = {\theta}_j^T(t) {\zeta}_{ij}(t) -w_{ij}(z) [{\theta}_j^T{\omega}_j](t),\label{XiID_MIMO_Diag}
\end{equation}
\begin{equation}
    {\zeta}_{ij}(t) = w_{ij}(z) [{\omega}_j](t). \label{ZetaID_MIMO_Diag}
\end{equation}
The signal $\xi_{ij}(t)$ has the desired property that $\xi_{ij}(t) = 0 $ for ${\theta}_j(t) = {\theta}_j$ being constant. However, the estimation error $\epsilon_i(t)$ depends on ${\theta}_j(t), j = 1,\ldots,m$.

We apply the same estimation error of the gradient algorithm based adaptive control scheme defined in \eqref{EstErrorID_MIMO} for the least-squares algorithm based adaptive control scheme, whose component $\epsilon_i(t)$ can also be expressed as 
\begin{equation}
\begin{split}
    \epsilon_i (t)  = e_{xi}(t) + {\theta}^T(t){\zeta}_i(t) - {\nu}_i(t),   
\end{split}
    \label{Eps_i4algorID_Diag}
\end{equation}
where 
\begin{equation}
    {\theta}(t) = [{\theta}_1^T(t),\ldots,{\theta}^T_m(t)]^T\in \mathbb{R}^{m(n+1)},
\label{BthetaDef}
\end{equation}
\begin{equation}
    {\theta}^* = [{\theta}^{*T}_1,\ldots,{\theta}^{*T}_m]^T\in \mathbb{R}^{m(n+1)},
    \label{Btheta_def}
\end{equation}
\begin{equation}
    {\zeta}_i(t) = [{\zeta}_{i1}^T(t),\ldots,{\zeta}_{im}^T(t)]^T\in \mathbb{R}^{m(n+1)},
\end{equation}
\begin{equation}
    {\nu}_i(t) = \sum_{j=1}^m w_{ij}(z)[{\theta}^T_j{\omega}_j](t) . 
\end{equation}

Substituting \eqref{exEleID_MIMO_Diag} - \eqref{ZetaID_MIMO_Diag} to the estimation error \eqref{Eps_i4algorID_Diag}, we also have 
\begin{equation}
\begin{split}
        {\epsilon}_i (t) &= 
         \big( {\theta}(t) - {\theta}^* \big)^T{\zeta}_i(t).
\end{split}
\label{Eps_i4Proof_diag}
\end{equation}

From \eqref{Eps_i4algorID_Diag}, in the vector form, we can express $\epsilon(t) =      
[\epsilon_1(t), \ldots, \epsilon_n(t)]^T$ as
\beq
\epsilon(t) = {\mu}(t) + {Z}^{T}(t) {\theta}(t),
\eeq
where ${\mu}(t) = [{\mu}_1(t), \ldots, {\mu}_n(t)]^T$ with
\beq
{\mu}_i(t) = e_{xi}(t) - {\nu}_{i}(t) ,
\label{muDefineID_Diag}
\eeq
and matrix ${Z}(t)$ is defined as 
\begin{equation}
    {Z}(t) =  [{\zeta}_1(t),\ldots,{\zeta}_n(t)] \in \mathbb{R}^{m(n+1)\times n}. 
    \label{ZDefineID_Diag}
\end{equation}

\subsection{Adaptive Law and Its Properties}
\label{Sec:ALDA}
In this subsection, the adaptive law using the least-squares algorithm to minimize the accumulated estimation error is developed. The adaptive law optimality and the system state tracking performance are analyzed.

\medskip
\textbf{Cost function}. Based on the the estimation error \eqref{Eps_i4algorID_Diag} with the signal $\mu_i(t)$ defined in \eqref{muDefineID_Diag}, the measured signal $y_i(t)$ in the cost function \eqref{Cost_GLS} is $-\mu_i(t)$ in this case. Thus, the cost function for the least-squares algorithm based adaptive control scheme is 
\begin{equation}
\begin{split}
    J({\theta}) 
    = & \; \frac{1}{2}\sum_{\tau=t_0}^{t-1} \frac{1}{\kappa}\sum_{i=1}^n({\mu}_i(\tau) + {\theta}^T(t){\zeta}_i(\tau) )^2 +\frac{1}{2}({\theta}(t)-{\theta}_0)^TP_0^{-1}({\theta}(t)-{\theta}_0),
\end{split}
\label{CostFuncMIMO}
\end{equation}
where $P_0=P_0^T>0$ and $\kappa>0$, which is a combination of a sum of squared estimation errors at many time instants with a penalty on the initial estimate ${\theta}(0) = {\theta}_0$ of ${\theta}^*$.

\medskip
\textbf{Adaptive law}. The adaptive law of the parameter estimate ${\theta}(t)$ is derived as  
\beq
\label{theta(t+1)_resID_Diag}
{\theta}(t+1) = {\theta}(t) - P(t-1){Z}(t)N^{-1}(t)  \epsilon(t),
\eeq
where
\begin{equation}
    N(t)= \kappa I + {Z}^T(t) P(t-1)
{Z}(t),
\label{mDefinition}
\end{equation}
and $P(t)$ is recursively calculated from 
\beq
\label{P(t+1)_resID_diag} 
P(t) = P(t-1) - P(t-1) {Z}(t)N^{-1}(t)  {Z}^T(t)
P(t-1),
\eeq
with ${\theta}(0) ={\theta}_0$ and $P(-1) = P_0 = P_0^T > 0$ chosen.

\medskip
{\textbf{Adaptive law optimality}}. The optimality of the adaptive law \eqref{theta(t+1)_resID_Diag} is depicted in the following property.
\begin{lem}
    The adaptive law \eqref{theta(t+1)_resID_Diag} minimizes the cost function \eqref{CostFuncMIMO}, the accumulated estimation error with a penalty on the initial estimate $\theta_0$, at every time instant $t$.
\end{lem}

\begin{proof}
The closed form of the optimal solution ${\theta}(t)$ for each time instant $t$ can be found by setting 
\begin{equation}
\begin{split}
    \frac{\partial J({\theta})}{\partial {\theta}}  = \sum_{\tau=t_0}^{t-1}\frac{1}{\kappa}\sum_{i=1}^n\left({\mu}_i(\tau) + {\theta}^T(t){\zeta}_i(\tau)  \right){\zeta}_i(\tau) + P^{-1}_0\big({\theta}(t)-{\theta}_0\big) &= 0.  \\
\end{split}
    \label{GradientJRawID_Diag}
\end{equation}
With the definition of ${Z}(t)$ in \eqref{ZDefineID_Diag}, we have 
\begin{equation}
\sum_{i=1}^n{\zeta}_i(\tau){\zeta}_i^T(\tau) = {Z}(\tau){Z}^T(\tau) ,  
\label{Z=sumID_Diag}
\end{equation}
\begin{equation}
    \sum_{i=1}^n {\mu}_i(\tau){\zeta}_i(\tau) = {Z}(\tau){\mu}(\tau).
    \label{Z=summuID_Diag}
\end{equation}

According to \eqref{Z=sumID_Diag}, \eqref{GradientJRawID_Diag} can be rewritten as 
\begin{equation}
    \frac{\partial J({\theta})}{\partial {\theta}}  =\big(P^{-1}_0+\sum_{\tau = 0}^{t-1}\frac{1}{\kappa} {Z}(\tau){Z}^T(\tau) \big){\theta}(t)+\sum_{\tau = t_0}^{t-1} \frac{1}{\kappa}\sum_{i=1}^n {\mu}_i(\tau){\zeta}_i(\tau)-P_0^{-1}{\theta}_0  = 0  .
    \label{GradientJZID_Diag}
\end{equation}

Based on \eqref{GradientJZID_Diag}, ${\theta}(t)$ can be expressed as 
\begin{equation}
    {\theta}(t) = P(t-1)\left(-\sum_{\tau = t_0}^{t-1} \frac{1}{\kappa} \sum_{i=1}^n {\mu}_i(\tau){\zeta}_i(\tau)+P_0^{-1}{\theta}_0 \right),
    \label{theta_rawID_diag}
\end{equation}
with
\begin{equation}
    P(t-1)=\left(P^{-1}_0+\sum_{\tau = t_0}^{t-1}\frac{1}{\kappa} {Z}(\tau){Z}^T(\tau)\right)^{-1}.
    \label{Pk_rawID_diag}
\end{equation}

According to \eqref{Pk_rawID_diag}, the recurrent expression of $P^{-1}(t)$ is 
\begin{equation}
\begin{split}
    P^{-1}(t) &= P^{-1}(t-1)+\frac{1}{\kappa} {Z}(t){Z}^T(t).\\
\end{split}
\label{Pk^-1}
\end{equation}
Then, we have the recursive expression  \eqref{P(t+1)_resID_diag} for $P(t)$:
\begin{equation}
    P(t) = P(t-1) - {P(t-1){Z}(t)({\kappa I +{Z}^T(t)P(t-1){Z}(t)})^{-1}{Z}^T(t)P(t-1)}, 
        \label{Pk_finalID_diag}
\end{equation}
by substituting $ A = P^{-1}(t-1),\, B = \frac{1}{\kappa}Z(t),\, C = Z^T(t)$
to the equation 
\begin{equation}
    (A + B C)^{-1} = A^{-1} - A^{-1} B(I + C A^{-1} B)^{-1} C A^{-1}. 
\end{equation}

\bigskip
With \eqref{Z=summuID_Diag}, \eqref{theta_rawID_diag} and \eqref{Pk_finalID_diag}, we have 
\begin{equation}
\begin{split}
        P^{-1}(t-1) {\theta}(t) &= P_0^{-1}{\theta}_0 - \sum_{\tau = t_0}^{t-1}\frac{1}{\kappa}\sum_{i=1}^n {\mu}_i(\tau)\zeta_i(\tau) = P_0^{-1}{\theta}_0 - \sum_{\tau = t_0}^{t-1}\frac{1}{\kappa} {Z}(\tau){\mu}(\tau).
\end{split}
\label{P-1ThetaID_diag}
\end{equation}

After combining \eqref{Eps_i4algorID_Diag}, \eqref{ZDefineID_Diag}, \eqref{theta_rawID_diag}, and \eqref{P-1ThetaID_diag}, we obtain 
\begin{equation}
\begin{split}
{\theta}(t+1) &= P(t)\left(-\sum_{\tau = t_0}^{t} \frac{1}{\kappa} \sum_{i=1}^n {\mu}_i(\tau){\zeta}_i(\tau)+P_0^{-1}{\theta}_0 \right)    \\
&=P(t)\left( -\frac{1}{\kappa} {Z}(t){\mu}(t) + P(t-1)^{-1}{\theta}(t) \right) \\
& = P(t)\left( -\frac{1}{\kappa} {Z}(t){\mu}(t) + \left(P^{-1}(t) - \frac{1}{\kappa} {Z}(t){Z}^T(t) \right ) {\theta}(t) \right)\\
& = {\theta}(t)-\frac{1}{\kappa} P(t) {Z}(t) \left({Z}^T(t) {\theta}(t)  +{\mu}(t)\right)\\
&= {\theta}(t)-\frac{1}{\kappa} P(t){Z}(t) \epsilon(t)\\
& = {\theta}(t)- P(t-1){Z}(t)\big(\kappa I + {Z}^T(t) P(t-1)
{Z}(t)\big)^{-1} \epsilon(t),
\end{split}
\label{Theta_finalID_diag}
\end{equation}
which is the adaptive law \eqref{theta(t+1)_resID_Diag}. 
This completes Lemma's proof.
\end{proof}

\medskip
{\textbf{Adaptive law properties}}. The adaptive law \eqref{theta(t+1)_resID_Diag} has the following properties.

\begin{lem}
\label{Lemma: MIMOStabIND}
    The adaptive law \eqref{theta(t+1)_resID_Diag} guarantees that \begin{itemize}
    \item[(i)] $P(t)=P^T(t)> 0 $, $\forall t\in\{0,1,2,\ldots \}$, and $P(t)$ is bounded;
    \item[(ii)] ${\theta}(t)$, $\epsilon^T(t)N^{-1}(t)\epsilon(t)$ and $\epsilon^T(t)\Bar{N}^{-1}(t)\epsilon(t)$ are bounded, where $\Bar{N}(t) = I + {Z}^T(t){Z}(t)$; and
    \item [(iii)] $N^{-\frac{1}{2}}(t)\epsilon(t)$, ${\theta}(t+1)-{\theta}(t) $ and $\Bar{N}^{-\frac{1}{2}}(t)\epsilon(t)$  belong to $L^2$.
\end{itemize}
\end{lem}

\begin{proof}
     (i) According to \eqref{Pk_rawID_diag} and \eqref{Pk^-1}, we know that $P^{-1}(t)$ is nondecreasing, which means $P^{-1}(t) = \big(P^{-1}(t)\big)^T\geq P^{-1}_0>0$. Thus, we have $P(t)=P^T(t)> 0 $, $\forall t\in\{0,1,2,\ldots \}$, and $P(t)$ is bounded.

\medskip

(ii) With 
    $\Tilde{{\theta}}(t) = {\theta}(t) - {\theta}^*,$
\eqref{Eps_i4Proof_diag}, \eqref{theta(t+1)_resID_Diag} and \eqref{P(t+1)_resID_diag}, we have 
\begin{equation}
\begin{split}
        \Tilde{{\theta}}(t+1) &= \Tilde{{\theta}}(t) - P(t-1){Z}(t)\big(\kappa I + {Z}^T(t) P(t-1){Z}(t)\big)^{-1}{Z}^T(t)\Tilde{{\theta}}(t) \\
        & = \left(I - P(t-1){Z}(t)\big(\kappa I + {Z}^T(t) P(t-1){Z}(t)\big)^{-1}{Z}^T(t)\right)\Tilde{{\theta}}(t)\\
        &= P(t)P^{-1}(t-1)\Tilde{{\theta}}(t) .
\end{split}
\label{ForThetaCov_diag}
\end{equation}

Consider the positive definite function 
\begin{equation}
    V(\Tilde{{\theta}},t) = \Tilde{{\theta}}^TP^{-1}(t-1)\Tilde{{\theta}}.
\end{equation}
 Then, the time increment of $V = V(\Tilde{{\theta}}(t),t)$, along \eqref{theta(t+1)_resID_Diag}, is 
\begin{equation}
    \begin{split}
        &\; V(\Tilde{{\theta}}(t+1),t+1) - V(\Tilde{{\theta}}(t),t)\\
        =&\; \Tilde{{\theta}}^T(t+1)P^{-1}(t)\Tilde{{\theta}}(t+1) - \Tilde{{\theta}}^T(t)P^{-1}(t-1)\Tilde{{\theta}}(t)\\
        =&\; (\Tilde{{\theta}}^T(t+1) - \Tilde{{\theta}}^T(t))P^{-1}(t-1)\Tilde{{\theta}}(t)\\
        =& \; -\epsilon^T(t)N^{-1}(t )\epsilon(t).
    \end{split}
    \label{DVraw_diag_p1}
\end{equation}
With the definition of $l^2$ norm of a vector signal $x$,  
\begin{equation}
    \|x\|_2 = \sqrt{x^Tx},
    \label{l2Lemma3.2}
\end{equation} 
 \eqref{DVraw_diag_p1} can be rewritten as 
\begin{equation}
\begin{split}
    &\; V(\Tilde{{\theta}}(t+1),t+1) - V(\Tilde{{\theta}}(t),t)
=- \left\|N^{-\frac{1}{2}}(t)\epsilon(t) \right\|_2^2\leq0,
\end{split}
    \label{DVraw_diag}
\end{equation}
which indicates $V(\Tilde{{\theta}},t) $ is bounded. 
With \eqref{Pk_rawID_diag}, we have 
\begin{equation}
\begin{split}
        V(\Tilde{{\theta}}(t),t) & = \Tilde{{\theta}}^T(t)P^{-1}_0\Tilde{{\theta}}(t) + \Tilde{{\theta}}^T(t)\left(\sum_{\tau = 0}^{t} \frac{1}{\kappa}{Z}(\tau){Z}^T(\tau) \right)\Tilde{{\theta}}(t)\\
        & = \Tilde{{\theta}}^T(t)P^{-1}_0\Tilde{{\theta}}(t) + \Tilde{{\theta}}^T(t)\left(\sum_{\tau = 0}^{t} \frac{1}{\kappa}\sum_{i=1}^n {\zeta}_i(\tau){\zeta}_i(\tau)^T \right)\Tilde{{\theta}}(t),
\end{split}
\end{equation}
which implies that $\Tilde{{\theta}}^T(t)P^{-1}_0\Tilde{{\theta}}(t)$ is bounded, and so are $\Tilde{{\theta}}(t)$ and ${\theta}(t)$. Since the time increment of $V(\Tilde{\theta},t)$ and $\epsilon^T(t)N^{-1}(t)\epsilon(t)$ belong to $L^1$ because of \eqref{DVraw_diag_p1}, we have, for this discrete-time case, that $\epsilon^T(t) N^{-1}(t)\epsilon(t)$ is bounded.

Based on $P_0 \geq P(t) = P^{T}(t)>0$ proved in (i), we have 
\begin{equation}
  \kappa I + {Z}^T(t) P(t-1)
{Z}(t) \leq a_0\left( I +{Z}^T(t) 
{Z}(t)  \right),
\label{NgeqaNb}
\end{equation}
where $a_0 = \max \{\kappa, \lambda_{max}(P_0)\}$. 

From \eqref{NgeqaNb} and the definitions of $N(t)$ and $\Bar{N}(t)$, $N(t)$ and $\Bar{N}(t)$ satisfy the following inequality  \begin{equation}
     N^{-1}(t) \geq \frac{1}{a_0} \Bar{N}^{-1}(t).
\end{equation}

With the boundedness of $\epsilon^T(t){N}^{-1}(t)\epsilon(t)$, the boundedness of $\epsilon^T(t)\Bar{N}^{-1}(t)\epsilon(t)$ is proved by
\begin{equation}
    \epsilon^T(t)\Bar{N}^{-1}(t)\epsilon(t) \leq {a_0}\epsilon^T(t){N}^{-1}(t)\epsilon(t) < \infty. 
    \label{boundepsNb-1eps}
\end{equation}

\medskip
(iii) With the boundedness of $V(\Tilde{\theta}(t))$ and \eqref{DVraw_diag_p1}, we have 
\begin{equation}
    \begin{split}
        \sum_{\tau = 0}^t\epsilon^T(\tau)N^{-1}(\tau)\epsilon(\tau) =\sum_{\tau = 0}^t \left\|N^{-\frac{1}{2}}(\tau)\epsilon(\tau) \right\|_2^2 = V(\Tilde{\theta}(0))- V(\Tilde{\theta}(t))\leq V(\Tilde{\theta}(0)),
    \end{split}
    \label{me_L2}
\end{equation}
that is, $N^{-\frac{1}{2}}(t)\epsilon(t) \in L^2$.

Combine \eqref{boundepsNb-1eps}, we have
\begin{equation}
    \left\|\Bar{N}^{-\frac{1}{2}}(\tau)\epsilon(\tau) \right\|_2^2 = \epsilon^T(t)\Bar{N}^{-1}(t)\epsilon(t)\leq {a_0}\epsilon^T(t){N}^{-1}(t)\epsilon(t) < \infty. 
   \label{emb_quad_trans}
\end{equation}
From \eqref{me_L2} and \eqref{emb_quad_trans}, it follows that $\Bar{N}^{-\frac{1}{2}}(t)\epsilon(t)\in L^2$.

To consider the $L^2$ property of $\theta(t+1)-\theta(t)$, we need to use the definition and some properties of the matrix norm $\vertiii{\cdot}$ that is ``induced" by a vector norm $||\cdot||$ introduced in \cite{HJ13}. 

\begin{defn} 
\label{defn: MatrixNorm}
Let $||\cdot||$ be a norm on $\mathbb{C}^n$. Define $\vertiii{\cdot}$ on $\mathbb{M}_{m\times n}$ by 
\begin{equation}
    \vertiii{A} = \max_{\|x\|=1} \|Ax\| .
\end{equation}
    
\end{defn}

\begin{prop}
\label{prop: MatrixNorm}
    The function $\vertiii{\cdot}$ defined in Definition \ref{defn: MatrixNorm} has the following properties:
    \begin{itemize}
        \item[(a)] $\vertiii{I} = 1$;
        \item[(b)] $\|Ay\|\leq \vertiii{A}\|y\|$ for any $A\in \mathbb{M}_{m\times n}$ and any $y\in \mathbb{C}^n$;
        \item[(c)]  $\vertiii{\cdot} $ is a matrix norm on $\mathbb{M}_{m\times n}$.
    \end{itemize}
\end{prop}

\begin{prop}
\label{prop: MatrixNormAxioms}
From $(c)$ of Proposition \ref{prop: MatrixNorm}, we know the function $\vertiii{\cdot}$ defined in Definition \ref{defn: MatrixNorm} satisfies the nonnegativity, positivity, homogeneity, triangle inequality, and submultiplicativity axioms. 
\end{prop}

From Proposition \ref{prop: MatrixNorm}, the $l^2$ vector norm induced matrix norm is $\vertiii{A}_2 = \max_{\|x\|_2=1} \|Ax\|_2 = \sigma_1(A) $, the largest singular value of $A$.

Then, we have the inequality about the $l^2$ norm of $\theta(t+1)-\theta(t)$: 
\begin{equation}
    \begin{split}
        & \;\|\theta(t+1)-\theta(t)\|_2\\ \leq&\; \left\|P_s(t-1)P_s^T(t-1)Z(t)\left(\kappa I + \left(P_s^T(t-1)Z(t)\right)^2\right)^{-\frac{1}{2}}\right\|_2\|{N^{-\frac{1}{2}}}(t)\epsilon(t)\|_2\\
        \leq &\; \vertiii{P_s(t-1)}_2 \|{N^{-\frac{1}{2}}}(t)\epsilon(t)\|_2\\
        &\;\vertiii{Q_{P_sZ}(t)\begin{bmatrix}
            \lambda_{P_sZ}^1(t)&&\\
            &\ddots&\\
            &&\lambda_{P_sZ}^n(t)
        \end{bmatrix}\cancel{Q^{-1}_{P_sZ}(t)}\cancel{Q_{P_sZ}(t)}
        \begin{bmatrix}
            \frac{1}{\sqrt{\kappa+\left(\lambda_{P_sZ}^1(t)\right)^2}}&&\\
            &\ddots&\\
            &&\frac{1}{\sqrt{\kappa+\left(\lambda_{P_sZ}^n(t)\right)^2}}
        \end{bmatrix}Q_{P_sZ}^{-1}(t)}_2
    \end{split}
    \label{DThetal2_temp}
\end{equation}
where $P(t-1) \triangleq P_s(t-1) P_s^T(t-1) $, $Q_{P_sZ}(t)$ is a square matrix whose columns are
the $n$ linearly independent eigenvectors of $P_s^T(t-1)Z(t)$, 
and $\lambda_{P_sZ}^i(t)$ denotes the $i$th eigenvalue of $P_s^T(t-1)Z(t)$. 
Based on the spectral decomposition,  introduced in \footnote{\url{https://en.wikipedia.org/wiki/Eigendecomposition_of_a_matrix}}, $Q_{P_sZ}(t)$ is always orthogonal such that $\vertiii{Q_{P_sZ}(t)}^2 = 1$. Then, we rewritten \eqref{DThetal2_temp} as
\begin{equation}
    \begin{split}
        \|\theta(t+1)-\theta(t)\|_2 \leq \vertiii{P_s(t-1)}_2  \max_{1\leq i \leq n}\left|\frac{\lambda_{P_sZ}^i(t)}{\sqrt{\kappa+\left(\lambda_{P_sZ}^i(t)\right)^2}}\right|\|N^{-\frac{1}{2}}(t)\epsilon(t)\|_2.
    \end{split}
    \label{DThetal2}
\end{equation}
From the boundedness of $P(t)$ and \eqref{DThetal2}, it follows that $\theta(t+1)-\theta(t)\in L^2$

This completes the lemma's proof.
\end{proof}

\medskip
 Property (iii) in Lemma \ref{Lemma: MIMOStabIND} implies:
\begin{equation}
\lim_{t\rightarrow\infty}\left(\kappa I + {Z}^T(t)P(t-1){Z}(t)\right)^{-\frac{1}{2}}\epsilon(t) = 0    
\end{equation}
\begin{equation}
    \lim_{t\rightarrow\infty} \left ( {\theta}(t+1)- {\theta}(t) \right ) = 0
\end{equation} in this discrete-time case.

\medskip
\textbf{Adaptive control system properties}. We can then establish the following desired adaptive control system properties:
 
\begin{thm}
\label{Theorem:TrackingPerformance}
    The adaptive controller \eqref{uID_new} with adaptive law \eqref{theta(t+1)_resID_Diag} ensures the
    boundedness of all the closed-loop system signals and the achievement of the tracking objective $\lim_{t\rightarrow\infty}({x}(t)-x_m(t)) = 0$. 
\end{thm}

With the control law \eqref{uID_new}, we have $\hat{x}(t) = x_m(t)$ so that $x(t)-x_m(t) = x(t)-\hat{x}(t)$. Then, the operator concept based proof in \cite{T_DTASTC} is also applicable to the above Theorem \ref{Theorem:TrackingPerformance} of this paper with the properties listed in Lemma \ref{Lemma: MIMOStabIND}. With the operator operating related definitions and propositions, the existence of a stable and strictly proper operator $T_0(z)$, such that $T_0(z)\left[\|x\|\right]$ is bounded, is proved. Then, the boundedness of all the closed-loop system signals is proved. Lastly, the tracking performance, $\lim_{t\rightarrow\infty}e_x(t)=0$, is deduced from estimation error \eqref{EstErrorID_MIMO}, the $L^2$ properties of $(\kappa I + {Z}^T(t)P(t-1){Z}(t))^{-\frac{1}{2}}\epsilon(t) $, $\xi_{ij}(t)$, and ${\theta}(t+1)-\theta(t)$, and the system signal boundedness mentioned before.    

Different from the proof of Theorem 3.1 in \cite{T_DTASTC}, the proof of the indirect adaptive control design for MIMO systems using the least-squares algorithm needs to consider the different control input $u(t)$ in \eqref{uID_new} and the different definitions of the estimation error $\epsilon(t)$ while proving the existence of the operator $T_0(z)$. The detailed proof is included in Appendix. 

\medskip
\textbf{Parameter projection}. To ensure the estimated $\Theta^{-1}_2(t)$ always exists to construct the control law \eqref{uID_new}, the parameter projection schemes in \cite{T_DTASTC} and \cite{t03} can be applied in the adaptive law \eqref{theta(t+1)_resID_Diag} to project the estimates $\theta_{2j}(t),\,j=1,\ldots,m$, away from 0. 

\setcounter{equation}{0}
\section{Robot System Modelling and Collision Avoidance}
\label{sec: RMandCA}
This section presents modeling details of the multi-robot system as well as the collision avoidance mechanism. In Section \ref{sec:RM}, a discrete-time dynamic model of the multiple mobile robot system, that is suitable to implement the adaptive control algorithm introduced in Section \ref{Sec:LSAL} for the tracking control purpose, is demonstrated. In Section \ref{sec:CA}, the collision avoidance mechanism during the tracking control process and its influence on both the nominal control design and the adaptive control design are discussed.

\subsection{Robot System Model and State Tracking Control Problem}
\label{sec:RM}
This subsection offers an introduction to the robot system model and explains the meaning of the state tracking control problem.

\bigskip
\textbf{Robot system model}.
The 3-mobile-robot system in \cite{zt23May} is revised and used for validating the effectiveness of our proposed adaptive state tracking control scheme, where the simplified discrete-time robot model with sample time $\Delta t $ interval is described as 
\bea
v_{i}(t+1) \ts = \ts v_{i}(t) + a_{i}(t)\, \Delta t \nn\\
r_{i}(t+1) \ts = \ts r_{i}(t) + v_{i}(t)\, \Delta t + 0.5\,
a_{i}(t)\,(\Delta t)^2,
\label{3r_robots_indiv_accer}
\eea
with $r_i(k)=\left ( 
\begin{matrix}
x_i(t)\\
y_i(t)
\end{matrix}
\right)$, 
$ v_i(t)=\left(
\begin{matrix}
v_i^x(t)\\
v_i^y(t)
\end{matrix}\right)$, and $a_i(t)=\left(
\begin{matrix}
a_i^x(t)\\
a_i^y(t)
\end{matrix}\right),i = 1,2,3$ denoting the position, velocity and acceleration of Robot$i$, respectively. To include system uncertainties like unknown robot mass and friction coefficient, based on Newton's law, we reform the robot model \eqref{3r_robots_indiv_accer} as
\bea
v_{i}(t+1) \ts = \ts v_{i}(t) + \frac{1}{m}\,\left(u_{i}(t) - b v_{i}(t)\right)\, \Delta t  \nn\\
r_{i}(t+1) \ts = \ts r_{i}(t) + v_{i}(t)\, \Delta t + \frac{1}{2m}\,
\left(u_{i}(t) - b v_{i}(t)\right)\,(\Delta t)^2,
\label{3r_robots_indiv}
\eea
where the robot mass, friction coefficient, and control input (traction force)
are denoted by $m$, $b$ and $u_i(t)=\left(
\begin{matrix}
u_i^x(t)\\
u_i^y(t)
\end{matrix}\right),i = 1,2,3$, respectively.

With the total state vector of Robot$i$, 
\beq
\begin{split}
    X_i(k) = & \;[x_{i}(t), y_{i}(t),v_{i}^x(t),
  v_{i}^y(t)]^T \in \mathbb{R}^{4},
\end{split}
\label{3r_z1}
\eeq
we can express (\ref{3r_robots_indiv}) as
\beq
X_i(t+1) = A X_i(t) + B U_i(t),
\label{3r_z1_update}
\eeq
where the control vector of Robot$i$ $U_i(t)$ is
\beq
U_i(t) = [u_{i}^x(t), u_{i}^y(t)]^T \in \mathbb{R}^{2},
\eeq
and $A \in \mathbb{R}^{4 \times 4}$ and $B \in \mathbb{R}^{4 \times 2}$ are some
matrices as the controlled system matrices which can be directly
obtained from the system models (\ref{3r_robots_indiv}):
\beq
A = \left [
\begin{matrix}
I_{2\times2}& \, \left(1- \frac{0.5b(\Delta t)^2}{m}\right)I_{2\times 2}\\
0_{2\times2}& \left(1-\frac{b\Delta t}{m}\right)I_{2\times 2}\\
\end{matrix}
\right],B = \left [
\begin{matrix}
\frac{0.5(\Delta t)^2}{m}\,I_{2\times 2}\\
\frac{\Delta t}{m}\,I_{2\times 2}
\end{matrix}
\right] ,
\label{AB_value}
\eeq
where $\Delta t $ is the sampling period. 

\bigskip
\textbf{State tracking control problem}. The control objective is to design a state feedback control signal $U_i(t)$ to ensure that all closed-loop system signals are bounded and the system state vector $X_i(t)$ of robot model \eqref{3r_robots_indiv} asymptotically tracks a reference state vector $X_{mi}(t)$ generated from a chosen reference model system
\begin{equation}
\begin{split}
    X_{mi}(t+1) = A_{mi}X_{mi}(t)+B_{mi} R_i(t),\, X_{mi}(t) \in \mathbb{R}^{4},\, R_i(t)\in \mathbb{R}^2,
\end{split}\label{RefModel_simu}
\end{equation}
where $A_{mi}\in \mathbb{R}^{4\times 4 }$ and $B_{mi} \in \mathbb{R}^{4\times  2}$ are some constant matrices, and $R_i(t)$ is chosen reference input signal for desired system response. The adaptive state tracking control aims to ensure robots track the trajectories designed by a selected reference model with different reference inputs. 

\begin{rem}
\label{Rem:referncemodelandtraj}
Reference models \eqref{RefModel_simu} with reference inputs $R_i(t)$ are selected based on desired trajectories. The desired trajectories should make distances between robots larger than 0. \hfill$\square$
\end{rem}

\subsection{Modified Control Input for Collision Avoidance}
\label{sec:CA}
Because the controlled plants usually need a period to finish the tracking objective, robots may collide with others in this period as the tracking errors have not converged to zero yet, though the reference models can ensure the desired robot positions are far enough from others. 
To avoid possible collisions while the robot model trying to track the reference model, we add an additional repulsive force $F_{ri}(t)$, generated by the artificial repulsive potential fields around the robots based on \cite{bc14}, \cite{kb91}, \cite{repulsivebook}, to the control input $U_i(t)$.

\bigskip
\textbf{Repulsive field}. For every robot, referred as Robot $i$, surrounding this robot, the numerical values of this field increase along the direction towards the robot until reaching a significantly large value at the edge of the robot. They decrease along the direction moving away from the robot until reaching zero. 
Mathematically, the repulsive field of Robot $i$ is expressed as 
\begin{equation}
    W_{i} = \left\{
    \begin{array}{ll}
        \frac{1}{2}\eta_i\left(\frac{1}{\gamma}-\frac{1}{\rho_0}\right)^2& \rho(r_i)\leq \gamma \\
         \frac{1}{2}\eta_i\left(\frac{1}{\rho(r_{i})}-\frac{1}{\rho_0}\right)^2& \gamma < \rho(r_{i})\leq\rho_0 \\
         0& \text{else}, 
    \end{array}\right.
    \label{RepulsiveField}
\end{equation}
$i =1,2,3$, where $\eta_i$ is a positive design constant, $\gamma$ is the radius of the robot, $\rho(r_i)$ is the distance from Robot $i$, and $\rho_0$ is a so-called safe distance between robots.

\bigskip 
\textbf{Repulsive force}. Since this is a 3-robot system, Robot $i,$ $i = 1,2,3$, experiences repulsive forces generated by the fields of Robot $j$ which are represented as $f_{wi\leftarrow j} = \left(
\begin{matrix}
f_{wi\leftarrow j}^x\\
f_{wi\leftarrow j}^y
\end{matrix}\right)$, $j = 1,2,3$, $j\neq i$. 
{The notation $f_{wi\leftarrow j}$ denotes the force exerted on Robot $i$ by the artificial field generated around Robot $j$ which is computed as the negative gradient of the field, $\nabla{W_j}$}:
\begin{equation}
    f_{wi\leftarrow j}= \left\{
    \begin{array}{ll}
   \eta_j\left(\frac{1}{\gamma}-\frac{1}{\rho_0}\right) \frac{\nabla_{r_{ij}}\rho(r_{ij})}{\gamma^2} & \rho(r_{ij})\leq \gamma \\
         \eta_j\left(\frac{1}{\rho(r_{ij})}-\frac{1}{\rho_0}\right) \frac{\nabla_{r_{ij}}\rho(r_{ij})}{\rho^2(r_{ij})} & \gamma<\rho(r_{ij})\leq\rho_0 \\
         0& \text{else}, 
    \end{array}\right.
    \label{RepulsiveForce}
\end{equation}
where $\nabla_{r_{ij}}\rho(r_{ij}) = \left[\frac{\partial \rho}{\partial x},\frac{\partial \rho}{\partial y}\right]^T$ denotes the gradient of the distance from Robot $i$ to Robot $j$ (the components in the x-direction and y-direction of the unit vector directed to Robot $i$). Unlike \eqref{RepulsiveField}, in \eqref{RepulsiveForce}, we use $\rho(r_{ij})$ to replace 
$\rho(r_{i})$ from \eqref{RepulsiveField}, to emphasize that $f_{wi\leftarrow j}$ is directly related to the distance between Robot $i$ and Robot $j$.

\bigskip
{\textbf{Collision avoidance property of repulsive force}} The collision avoidance performance is guaranteed by the conservation of energy. We summarize the performance of the repulsive force $f_{wi}$ in \eqref{RepulsiveForce} in the followed proposition. 
\begin{prop}
    When robots are only controlled by the repulsive fore $f_{wi\leftarrow j}$ ($\alpha_i(t) = 0$, $\forall t \in \mathbb{N}$), no collisions between robots will happen. 
\end{prop}

\begin{proof}
The energy applied to Robot $i$ by $U_i(t)$ is distributed into three components: kinetic energy of Robot $i$, potential energies $W_j(\rho(r_{ij}))$ stored as Robot $i$ approaches Robot $j$, $ j \neq i$, and the energy used to overcome friction.

With $v_{max}$ denoting constant as the maximum value of the robots' speed and combining with the definition of the field \eqref{RepulsiveField}, one can find constants $\rho_0$ and $\eta_j$ in \eqref{RepulsiveForce} such that 
\begin{align}
    W_{j}(\rho_{min}) \geq W_{j}&(\rho_{0}) + \frac{1}{2} m \|v_{max}\|_2^2,
    \label{W_CE}
\end{align}
where $\rho_{min}>0$ denotes a pre-defined minimal distance between robots, $m$ denotes the mass of the robots and $\Delta E > 0$ denotes the tolerance energy for tracking while the robot avoids collisions. 

When robots are only controlled by the repulsive fore $f_{wi\leftarrow j}$, \eqref{3r_robots_indiv} and \eqref{RepulsiveField}-\eqref{RepulsiveForce} imply
\begin{equation}
   \Delta E_{i\leftarrow j}(t) =  W_{j}(\rho_{min}) - W_{j}(\rho_{ij}(t))\geq 0, \forall t \in \mathbb{N},
    \label{WRhoijGeqRhomin}
\end{equation}
i.e., 
\begin{equation}
    \rho_{ij}(t) \geq \rho_{min},\,\forall t \in \mathbb{N}.
    \label{RhoijGeqRhomin}
\end{equation}

The collision avoidance property of the repulsive force is proved by \eqref{RhoijGeqRhomin}. 
\end{proof}

\bigskip
\textbf{Repulsive force applied on robot model}. To achieve our target, collision avoidance, we modify the control input $U_i(t)$ in \eqref{3r_z1_update} as 
\begin{equation}
   U_i(t) = F_{ri}(t) + \alpha_i(t) U_{oi}(t),
    \label{u_mix}
\end{equation}
where \begin{equation}
    \begin{split}
    F_{ri}(t) = \sum_{j\neq i} f_{wi\leftarrow j}(t) \in \mathbb{R}^{2}
\end{split}
\label{FrDef}
\end{equation}
denotes the repulsive force based on the artificial potential repulsive field $W_j, \, j \neq i$, $\alpha_i(t)$ is a coefficient designed to maintain tracking to some extent while avoiding collisions
and \begin{equation}
    U_{oi}(t) = {\Theta_{2i}^{-1}(t)} \left(\Theta_{1i}^T(t)X_i(t) + R_i(t)\right)
    \label{TCS}
\end{equation}
denotes the tracking control signal.

\begin{figure}[t]
    \centering
    \includegraphics[width=0.7\textwidth]{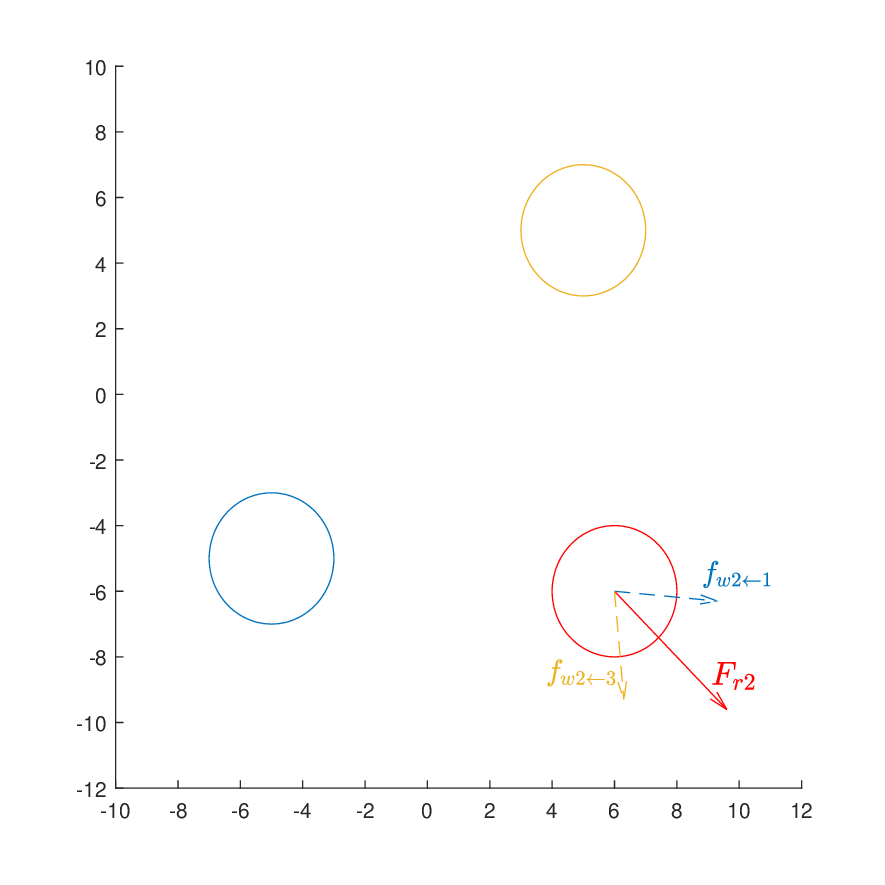}
    \caption{Schematic diagram of robot positions. }
    \label{Fig: schematic diagram of robot positions}
\end{figure}

As one robot may influenced by more than one robot, the repulsive force acts on every robot is the resultant force of all the field forces $f_{wi\leftarrow j}$ from the artificial field of other robots. We explain that with an example. The schematic diagram of robot positions is shown in Figure \ref{Fig: schematic diagram of robot positions}, Robot 1, Robot 2, and Robot 3 are depicted as three circles in blue, red, and yellow, respectively. For Robot 2 (the red circle), it not only undergoes the field force $f_{w2\leftarrow1}$ from Robot 1 (the blue circle), but also $f_{w2\leftarrow3}$ from Robot 3 (the yellow circle). Thus, the repulsive force of Robot2, $F_{r2}  $, is the vector sum of $f_{w2\leftarrow1}$ and $f_{w2\leftarrow3}$: 
\begin{equation}
    F_{r2}   
= \left(
\begin{matrix}
f_{w2\leftarrow1}^x + f_{w2\leftarrow3}^x \\
f_{w2\leftarrow1}^y + f_{w2\leftarrow3}^y
\end{matrix}\right).
\end{equation}

\smallskip
\begin{rem}
\label{Rem:SuspensionAdaptation}
    To avoid the influence of the estimation error and $Z(t) $ computation led by $F_{ri}(t)$, we will suspend the parameter adaptation once $F_{ri}(t)\neq0$. \hfill$\square$ 
\end{rem}

\bigskip
\textbf{Collision avoidance property of the modified control input}. 
Coefficient $\alpha_i(t)$ in \eqref{u_mix} must ensure that the additional energy contributed by the component of $U_{oi}(t)$ in the direction toward Robot $j$ does not violate the collision-free criteria \eqref{WRhoijGeqRhomin} at every sampling period. Thus, the maximum extra energy introduced by $U_{oi}(t)$ during each sampling period is
\begin{equation}
  E_{oi\rightarrow j}(t) =   -\frac{U_{oi}(t)\cdot f_{wi\leftarrow j}(t) }{\|f_{wi\leftarrow j}(t) \|_2}v_{max}\Delta t
    \label{ExtraEnergyMax}
\end{equation}
when $\|f_{wi\leftarrow j}(t) \|_2 \neq 0$, i.e., $\rho_{ij} (t)> \rho_0$. 

To ensure $U_{oi}(t)$ does not violate the collision-free criteria between Robot $i$ and all other robots, $\alpha_i(t)$ is designed as 
\begin{equation}
    \alpha_i(t) = \min\{\alpha_{i\leftarrow j }(t) \mid j \in \{1, 2, 3\}, j\neq i\},
    \label{alphaDesign}
\end{equation}
where 
\begin{align}
    \alpha_{i\leftarrow j }(t)  = 
         \left\{\begin{array}{ll}
         \min\left(\frac{\beta\Delta E_{i\leftarrow j}(t)}{E_{oi\rightarrow j}(t)}, 1\right) &  E_{oi\rightarrow j}(t) > 0 \\ 
        1 & \text{else},
    \end{array}\right.
    \label{alphainDifferentDirection}
\end{align}
with $\beta\in [0,1)$ denoting a design parameter.

\smallskip
We present the collision avoidance property of \eqref{u_mix} with coefficient $\alpha_i(t)$ designed in \eqref{alphaDesign} in the following proposition. Since the design principle of coefficient $\alpha_i(t)$ in \eqref{u_mix} is to ensure inequality \eqref{WRhoijGeqRhomin} holds, we can summarize the collision avoidance property in the following proposition.

\begin{prop}
\label{prop:repulsiveCA}
    \textit{The modified control input \eqref{u_mix} with the coefficient $\alpha_i(t)$ designed as \eqref{alphaDesign} guarantees that the distances between robots are larger than zero.}
\end{prop}
\begin{proof}
Based on the design principle of coefficient $\alpha_i(t)$ in \eqref{ExtraEnergyMax}-\eqref{alphainDifferentDirection}, 
the modified control design \( U_i(t), \) $i=1,2,3,$ in \eqref{u_mix} ensures that inequality introduced in \eqref{WRhoijGeqRhomin},
$W_{j}(\rho_{min}) - W_{j}(\rho_{ij}(t))\geq 0, \forall t \in \mathbb{N},$ where $ i,j \in \{1, 2, 3\}, j\neq i$,
always holds, which implies 
the guarantee of the collision avoidance property \(\rho_{ij}(t) \geq \rho_{min} > 0,\,\forall t \in \mathbb{N} \). 
\end{proof}

\bigskip
\textbf{Tracking property of the modified control input}. The tracking property is summarized in the following proposition. 

\begin{prop}
     \textit{The modified control input \eqref{u_mix} with the coefficient $\alpha_i(t)$ in \eqref{alphaDesign} guarantees that \begin{itemize}
    \item[(i)] $F_{ri}(t)$ converges to zero, i.e., $\lim_{t\rightarrow \infty} F_{ri}(t) = 0$;
    \item[(ii)] $F_{ri}(t)$, $U_i(t)$ and $X_i(t)$ are bounded; and
    \item [(iii)] the tracking objective $\lim_{t\rightarrow\infty}({X}_i(t)-X_{mi}(t)) = 0$ is achieved.
\end{itemize}}
\end{prop}

\begin{proof}
    (i) Because the reference models $X_{mi}(t) $ are designed based on desired trajectories making distances between robots are larger than 0 from Remark \ref{Rem:referncemodelandtraj}, we have $\lim_{t\rightarrow\infty}E_{oi\rightarrow j}(t) \leq 0$, which implies $\lim_{t\rightarrow \infty} F_{ri}(t) = 0$. 

\smallskip
(ii) Based on \eqref{WRhoijGeqRhomin}, we have $W_{i\leftarrow j}(\rho_{ij}) < W_{i\leftarrow j}(\rho_{min}),$
which implies 
$$\|f_{wi\leftarrow j}(\rho_{ij}(t)) \|< \|f_{wi\leftarrow j}(\rho_{min})\|.$$
With the definition of the collision avoidance signal $F_{ri}(t)$ applied on Robot $i$ in \eqref{FrDef}, $F_{ri}(t)$ is bounded. 

\smallskip
From Theorem \ref{Theorem:TrackingPerformance} and the suspension of parameter adaptation mentioned in Remark \ref{Rem:SuspensionAdaptation}, $\Theta_{1i}(t)$ and $\Theta_{2i}(t)$ are bounded and the elements of $\Theta_2(t)$ are projected away from 0. 
Because of the existence of $\alpha_i(t)$ for restricting the extra energy of the system and the boundedness of $F_{ri}(t)$, $U_i(t) $ is bounded. 

Since the robot system \eqref{3r_z1_update} is bounded-input bounded-output (BIBO) stable, the boundedness of the robot states $X_{i}(t)$ is guaranteed by the boundedness of $U_i(t)$.

\smallskip
(iii) According to (i): $\lim_{t\rightarrow \infty} F_{ri}(t) = 0$, the control process will become a normal adaptive state feedback tracking control process without the influence of the repulsive force $F_{ri}(t)$ at last, which means that the tracking objective $\lim_{t\rightarrow\infty}\left(X_i(t)-X_{mi}(t)\right) =0$ can always be achieved.
\end{proof}

\bigskip
\begin{rem}
    {At some specific points, the net force from the repulsive fields of different robots on a particular robot may be zero, the collision avoidance mechanism cannot work. However, since this particular robot still receives input for tracking, it will move in a specific direction toward its desired trajectory. As soon as it moves, the conditions causing the issue are broken, so it won't impact system performance.\hfill$\square$}
\end{rem}

\setcounter{equation}{0}
\section{Simulation Study}
\label{sec: Simulation}
This section presents the simulation study to evaluate the performance of the developed adaptive state feedback control scheme. In Section \ref{sec: SS(3r)}, the details of 3-mobile-robot system, which is chosen to be the controlled plant \eqref{Plant}, and a simulation case, that satisfies the constraints on the physical properties of the robot, is presented. Simulation results are demonstrated in Section \ref{sec:SR}.

\subsection{Simulation System}
\label{sec: SS(3r)}
In this subsection, the simulation system constructed by three mobile robots is detailed. 
All the system states and system inputs should respect the physical constraints.

\medskip
\subsubsection{\textbf{A 3-Mobile-Robot System Model and Parameter Values}}
 The 3-mobile-robot system model for the simulation study is introduced as \eqref{3r_z1_update} of Section \ref{sec:RM}: 
 \beq
X_i(t+1) = A X_i(t) + B U_i(t), \;i = 1,2,3,
\label{3r_z1_update_simu}
\eeq
where the control vector of Robot$i$ $U_i(t)=[u_{i}^x(t), u_{i}^y(t)]^T$ 
and $A = \left [
\begin{matrix}
I_{2\times2}& \, \left(1- \frac{0.5b(\Delta t)^2}{m}\right)I_{2\times 2}\\
0_{2\times2}& \left(1-\frac{b\Delta t}{m}\right)I_{2\times 2}\\
\end{matrix}
\right],B = \left [
\begin{matrix}
\frac{0.5(\Delta t)^2}{m}\,I_{2\times 2}\\
\frac{\Delta t}{m}\,I_{2\times 2}
\end{matrix}
\right] ,$
with $\Delta t = 0.05 $ sec being the sampling period. 

\medskip
The robot mass in \eqref{3r_z1_update} is configured based on TurtleBot\footnote{In this paper, we consider a TurtleBot carrying additional loads, bringing its total mass to 18 kg. The feature details of TurtleBot are available in \url{https://emanual.robotis.com/docs/en/platform/turtlebot3/features/}.}, $ m = 18 \,{\rm kg}.$ From \cite{CM09FrictionValue}, the friction coefficient in \eqref{3r_z1_update} is set as $b = 4\,  {\rm N}\cdot{\rm sec/m}.$

\medskip
\subsubsection{\textbf{Reference Model Systems}}
The details of the robot reference models, the corresponding parameter values, and the reference input components are elaborated below.

\medskip
\textbf{Reference model}.
The matrices $A_{mi}$ and $B_{mi}$ of the reference model $X_{mi}(t+1)= A_{mi}X_{mi}(t) + B_{mi} R_i(t)$, are selected as 
\begin{equation}
\begin{split}
     A_{mi} = \begin{bmatrix}
        0.9999 I_{2\times2}  &0.9997I_{2\times2}\\
        -0.0028 I_{2\times2} & 0.775 I_{2\times2}\\
    \end{bmatrix}, 
    B_{mi} = \begin{bmatrix}
        -0.0007 I_{2\times2}\\
        -0.0278 I_{2\times2}\\
    \end{bmatrix},\;
\end{split}
i=1,2,3
    \label{AmBmValue}
\end{equation}
All the eigenvalues of $A_{mi}$ are inside the unit circle of the complex plane: $   0.9868$, $0.7881$. 

Based on \eqref{NewMatCond_MIMO} and \eqref{AmBmValue}, the matching parameters $\Theta_{1i}^*$ and $\Theta_{2i}^*$ are 
\begin{equation}
\begin{split}
    &\Theta_{1i}^{*}  =  \begin{bmatrix}
        0.1I_{2\times2}\;0.1I_{2\times2}
    \end{bmatrix}^T,\\ &\Theta_{2i}^* = -0.01 I_{{2\times 2}},\,i=1,2,3.
\end{split}
    \label{CaseIT}
\end{equation}

\medskip
Based on the above reference model, we consider making Robot 2 and Robot 3 move as two concentric circles with different radii and making Robot 1 stop at the $(0,0)$ point.

Then, reference inputs $R_i(t)$ in \eqref{RefModel_simu} are selected as $R_1(t) = 0.2[-\sin(\frac{\pi t}{2000}),\cos(\frac{\pi t}{2000})]^T$, $R_2(t) =0.375 [\sin(\frac{\pi t}{2000}),-\cos(\frac{\pi t}{2000})]^T$, and $R_3(t) = [0,0]^T$.

\medskip
\textbf{Initial value settings}.
For all the robots, every component of the initial parameter vector ${\theta}_{0}$ in \eqref{theta(t+1)_resID_Diag} is set as $0.625\theta^*$. 

The other simulation parameters are selected as $P_0 = I_{m(n+1)}$, $\kappa =0.00001$ for every individual robot. The initial states are selected as $X_1(0) = [0,    0    , 0,0]^T$, $X_2(0) = [0,    1.52   , 0,0]^T$, and $X_3(0) = [0.5, -1    , 0,0]^T$. 
The initial states of the reference model are set as $X_{mi}(0) = X_i(0), i= 1,2,3$.

\medskip
\textbf{Repulsive parameter settings}. Through our testing, we select $\beta = 0.9$, $\eta_i = 4.5$ and $\rho_0 = 0.36$ to be the parameter values of the repulsive force in \eqref{RepulsiveForce}.

\subsection{Simulation Results}
\label{sec:SR}

Simulation results for the proposed adaptive control algorithm with the repulsive force are displayed in this subsection.

\begin{figure}[H]
\begin{center}
\includegraphics[width=0.78\textwidth]{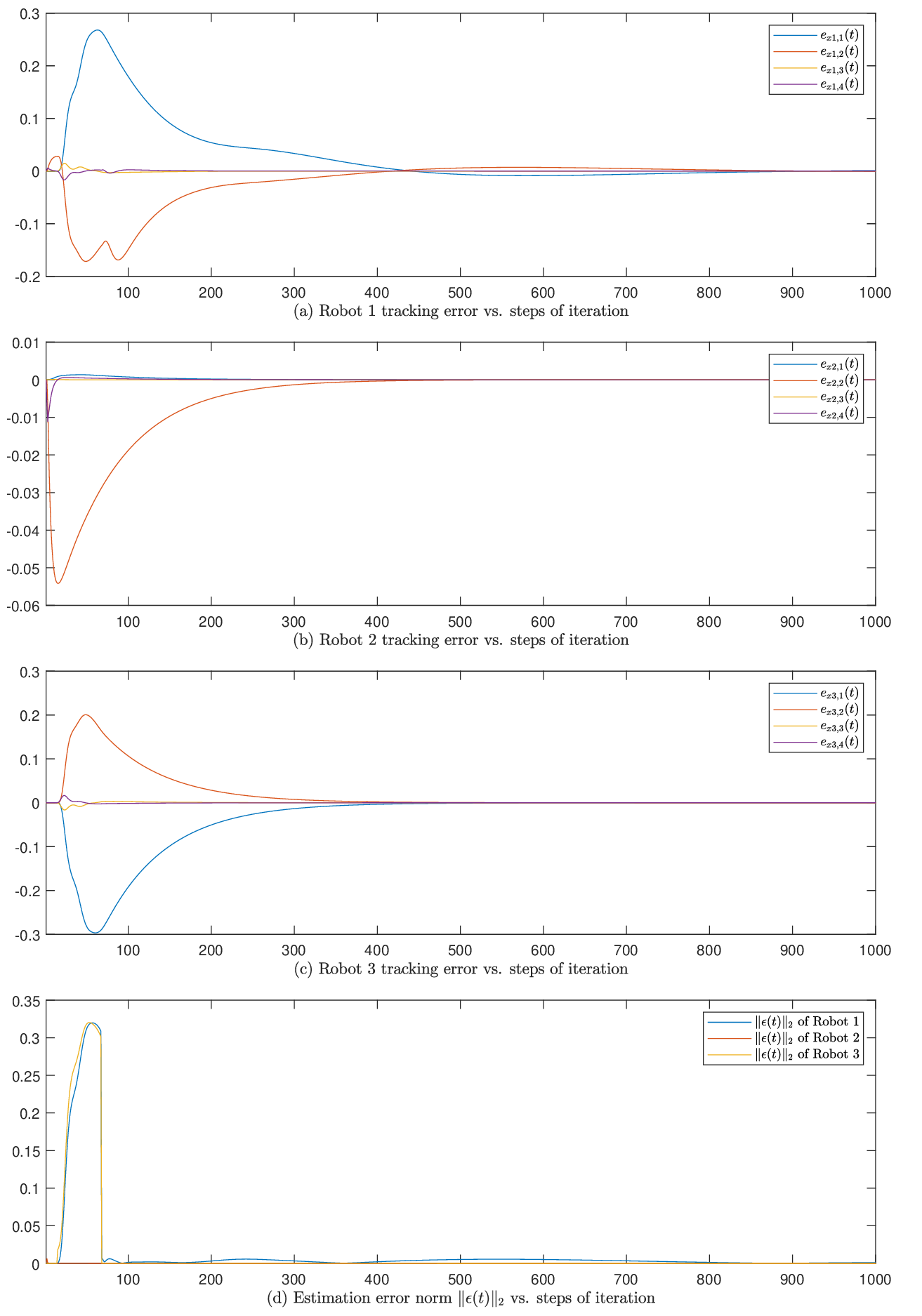}
\vspace{-0.3cm}
\caption{System response for the control input containing $F_{ri}(t)$. (a) Robot 1 tracking error components vs. steps of iteration. (b) Robot 2 tracking error components vs. steps of iteration. (c) Robot 3 tracking error components vs. steps of iteration. (d) Estimation error norm $\|\epsilon(t)\|_2$ vs. steps of iteration}
\label{Res_C1E}
\end{center}
\vspace{-0.5cm}
\end{figure}

According to Fig. \ref{Res_C1E} the closed-loop system is stable and all the components of the tracking error converge to zero asymptotically. Notably, the estimation error norm converging to zero asymptotically verifies the optimality of the adaptive laws, which minimizes the accumulated estimation error. 

From Fig. \ref{Res_C1T}, all the robots can move as the designed trajectories
controlled by either the control input containing $F_{ri}(t)$ or not. All the obvious direction changes of Robot 1 and 3 indicate the robots do make some efforts to avoid collisions\footnote{The animations of the robot motions can be found in \url{https://sites.google.com/virginia.edu/zhao-robotmotion-animation}.}.

In Fig. \ref{Res_C1D}, we compare the shortest distance between robots controlled by the modified control input \eqref{u_mix} and by the tracking control signal \eqref{TCS} only to see the collision avoidance performance. From the comparison chart, we observe that the blue curve always remains above zero, while some values of the red curve fall below zero. This indicates that our modified control design \eqref{u_mix} successfully avoids collisions.

\begin{figure}[t]
\begin{center}
\includegraphics[width=0.45\textwidth]{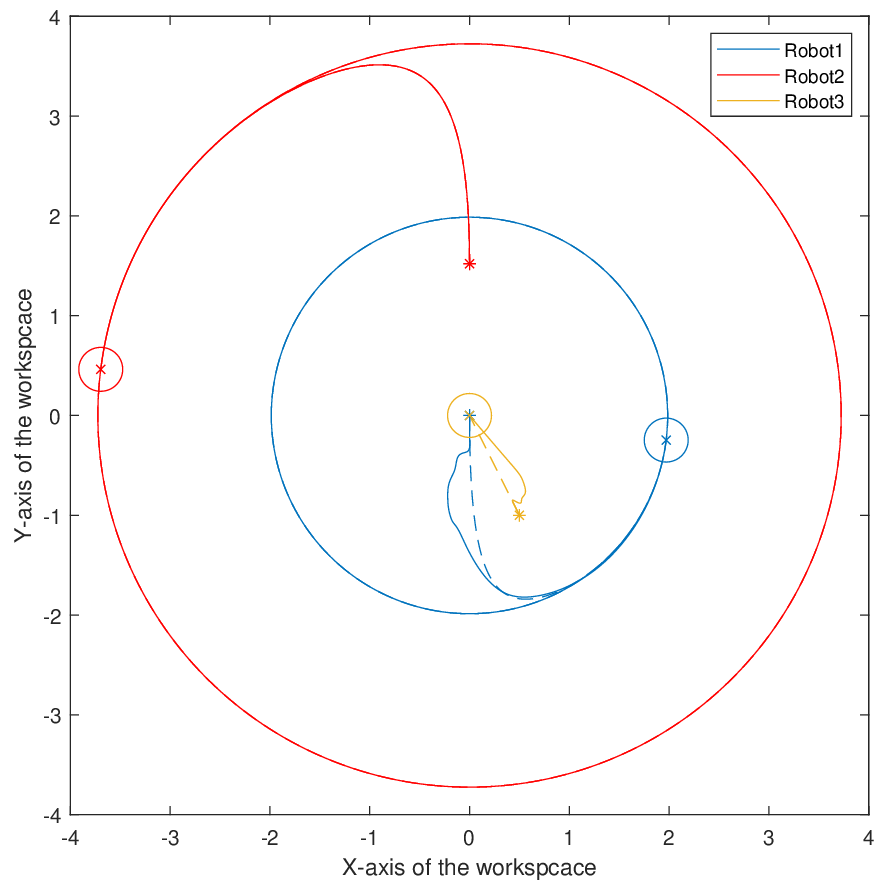}
\vspace{-0.3cm}
\caption{Robot trajectories controlled by modified control signal $U_{i}(t)$ (solid lines) or tracking control signal $U_{oi}(t)$ (dashed lines), with asterisk markers ($*$) indicating the central points of the robots' initial positions and circles enclosing colored product markers ($\times$) representing their final positions.}
\label{Res_C1T}
\end{center}
\vspace{-0.8cm}
\end{figure}

After examining the control inputs in the simulation process, all the control inputs are within $(-3.2,4.4)$, where the unit of input is Newton. This is reasonable because it implies that, based on Newton's laws, the robot's acceleration varies approximately within $(-0.18, 0.24)$, with the unit m/s$^2$.

From Fig. \ref{Res_C2G}, we find that the system applying the least-squares algorithm has a smaller absolute value of tracking error and a faster convergence speed of the tracking error.

With the above simulation study, both the state tracking performance and the collision avoidance performance of the proposed least-squares based adaptive control scheme are verified.

\begin{figure}[t]
\begin{center}
\includegraphics[width=0.78\textwidth]{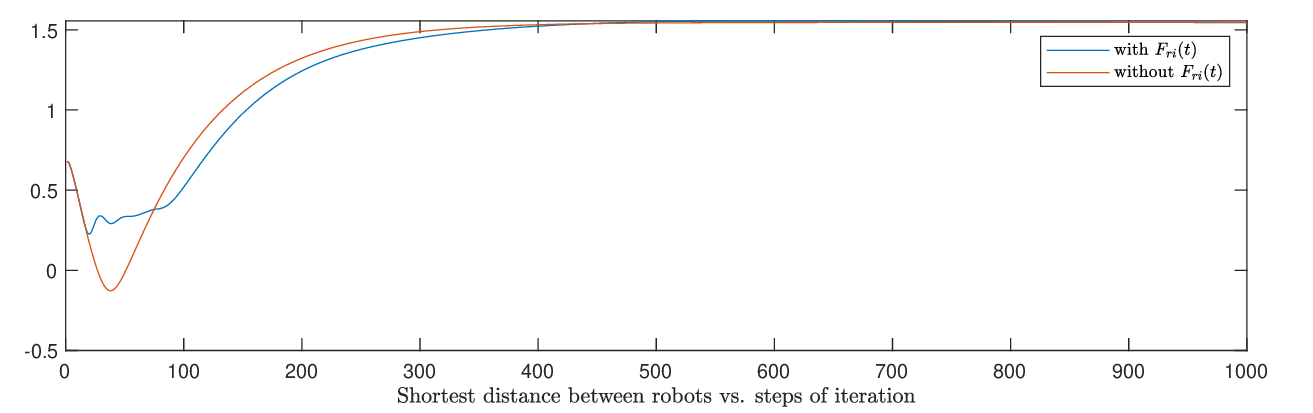}
\vspace{-0.3cm}
\caption{Shortest distance between robots comparison between the control input with or without collision avoidance mechanism.}
\label{Res_C1D}
\end{center}
\vspace{-0.5cm}
\end{figure} 

\begin{figure}[t]
\begin{center}
\includegraphics[width=0.78\textwidth]{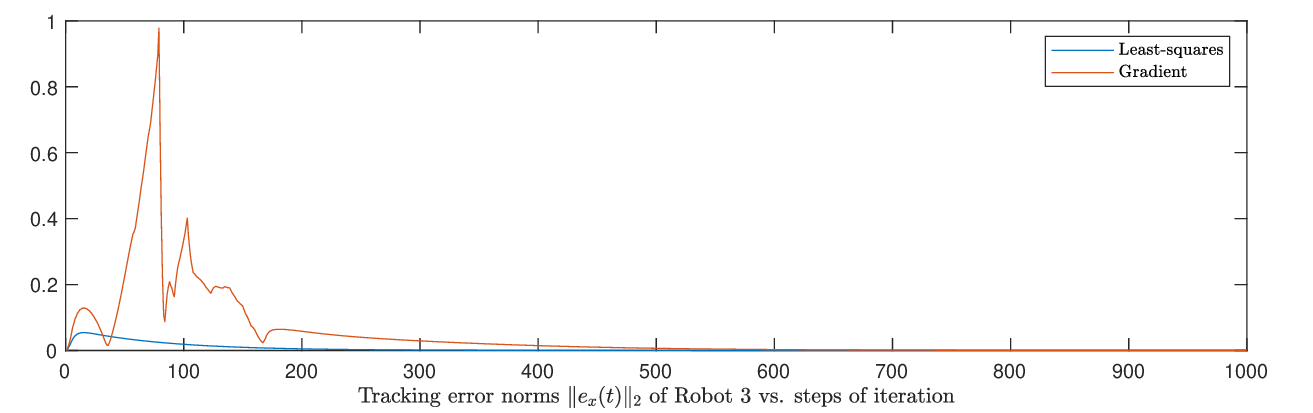}
\vspace{-0.3cm}
\caption{Tracking error norm comparison between the control scheme using least-squares or gradient algorithm (gain parameter of the gradient adaptive law is set as 1.9).}
\label{Res_C2G}
\end{center}
\vspace{-0.9cm}
\end{figure}

\section{Conclusions}

In this paper, we developed a new least-squares algorithm based adaptive
control scheme to solve the long-standing discrete-time indirect
adaptive state tracking control problem, which can ensure that the
state vector $x(t)$ of an unknown MIMO system: $x(t+1) = Ax(t)+Bu(t)$,
asymptotically tracks the state vector $x_m(t)$ of a chosen stable
reference model system. We also applied the developed adaptive control
scheme to solve a multiple mobile robot position and velocity tracking
control problem. Hence, this paper offered new solutions to two technical
problems: the stable adaptive law derivation for the discrete-time
MIMO adaptive state tracking control scheme based on the accumulative
parameterized estimation error minimization, and the design and
simulation of the least-squares adaptive multiple robot control scheme
with additional repulsive forces for both asymptotic trajectory tracking
and collision avoidance, whose desired performance is illustrated by
simulation results.
In the future, adaptive state tracking control techniques can be considered to deal with nonlinearities in robot systems.

\section*{Acknowledgment}

The authors would like to thank the financial support from a Ford University Research Program grant and the collaboration and help from Dr. Suzhou Huang and Dr. Qi Dai of Ford Motor Company for this research.

\renewcommand{\theequation}{A.\arabic{equation}}
\renewcommand{\theprop}{A.\arabic{prop}}
\renewcommand{\thedefn}{A.\arabic{defn}}
\setcounter{equation}{0}
\setcounter{defn}{0}
\setcounter{prop}{0}

\section*{Appendix: Proof of Theorem \ref{Theorem:TrackingPerformance}}
\label{Sec: App}

From $\hat{x}(t) = x_m(t) $ and \eqref{EstErrorID_MIMO}, we have 
\begin{equation}
    x(t) = {x_m}(t) - \epsilon(t) + \begin{bmatrix}
    \sum_{j=1}^m \xi_{1j}(t)\\
    \vdots\\
    \sum_{j=1}^m \xi_{nj}(t)
    \end{bmatrix}
    .  
    \label{x4bound}
\end{equation}

To prove the boundedness of $x(t)$, we need to use the definition and some properties of the matrix norm $\vertiii{\cdot}$ that is ``induced" by a vector norm $||\cdot||$ defined in Definition \ref{defn: MatrixNorm}.

Denote the $l^1$ norm of $x(t) $ as 
\begin{equation}
     \|{x}(t)\|  = {|{x}_1(t) |+ \cdots + |{x}_n(t)|},
    \label{xboundraw}
    \end{equation}
    from Proposition \ref{prop: MatrixNorm}, the $l^1$ vector norm induced matrix norm is $\vertiii{A}_1 = \max_{\|x\|_1=1} \|Ax\|_2 = \max_{1\leq i\leq n} \sum_{i=1}^n |a_{ij}| $, the maximum of the column sums of $A = [a_{i j}] \in M_{m\times n}$.

    From \eqref{x4bound} and triangle inequality, we have 
    \begin{equation}
    \begin{split}
        \|{x}(t)\| \leq&\, \|x_m(t)\| + \vertiii{\big(\kappa I + {Z}^T(t)P(t-1){Z}(t)\big)^{\frac{1}{2}}} \left\|N^{-\frac{1}{2}}(t)\epsilon(t)\right\| + {   \sum_{i=1}^n \sum_{j=1}^m  |\xi_{ij}(t)|},
    \end{split}
    \label{xl1bound}
    \end{equation}
    where $N^{-\frac{1}{2}}(t)\epsilon(t) \in L^2 \cap L^{\infty}$ from Lemma \ref{Lemma: MIMOStabIND}. 

According to \eqref{NgeqaNb}, we have 
\begin{equation}
    \vertiii{(\kappa I + {Z}^T(t)P(t-1){Z}(t))^{\frac{1}{2}}} \leq \vertiii{a_0 \left( I + {Z}^T(t){Z}(t)\right)^{\frac{1}{2}}}.
    \label{normNgeqaNb}
\end{equation}
For positive definite matrix $\Bar{N}(t) = I + {Z}^T(t){Z}(t)$,
we have orthogonal matrix $Q_m$ such that $\Bar{N}^{-1}(t) = Q_m(t)\Lambda_m(t)Q_m^T(t)$, with $\Lambda_m(t) = \diag(\lambda_{m}^1(t),\ldots,\lambda_{m}^n(t))$ containing all the eigenvalues $\lambda_{m}^i(t),\,i = 1,\ldots,n$ of matrix $N^{-1}(t)$. Then, ${Z}^T(t){Z}(t)$ can be rewritten as 
    \begin{equation}
       \left(Z^T(t)Z(t)\right)^{\frac{1}{2}} = Q_m(t)\Lambda^{-\frac{1}{2}}_m(t)Q_m^T(t) , 
    \end{equation}
    which means the eigenvalues of $\Bar{N}(t) = I + {Z}^T(t){Z}(t)$ are 
    $\frac{1}{\sqrt{\lambda_{m}^i}}(t),\,i = 1,\ldots,n$. 
    
    With $\lambda_{Z^TZ}^i(t)$ denoting the $i$th eigenvalue of $Z^T(t)Z(t)$ and the definition of $Z(t)$ in \eqref{ZDefineID_Diag}, we have the following inequality about the upper-bound of the eigenvalue of $\Bar{N}^{\frac{1}{2}}(t)$:
    \begin{equation}
       \lambda_m^i(t) = \sqrt{1 + \lambda_{Z^TZ}^i(t)} \leq   1 +  \sqrt{ \sum_{i=1}^n \zeta_i^T(t)\zeta_i(t)} 
    \end{equation}
    Then, based on \eqref{normNgeqaNb}, we have 
    \begin{equation}
    \begin{split}
        \vertiii{(\kappa I + {Z}^T(t)P(t-1){Z}(t))^{\frac{1}{2}}} 
        &\leq a_0\left(1 + \sqrt{ \sum_{i=1}^n \zeta_i^T(t)\zeta_i(t)}  \right) \, \vertiii{Q_m}^2\\
        &\leq a_0\left( 1 +  \sum_{i=1 }^n\|\zeta_i\|\right) \vertiii{Q_m}^2
        ,
    \end{split}
    \label{UpperboundNbar}
    \end{equation}
    where $\vertiii{Q_m}^2$ is bounded as every column of the orthogonal matrix $Q$ is a direction vector with a finite dimension.

For further proof for Theorem \ref{Theorem:TrackingPerformance}, we introduce the following definitions and propositions about basic operator concepts from \cite{T_DTASTC}. A linear operator $T(z,t)$ is applied  to indicate the relationship between the input signal $u(t)$ and the output signal $y(t)$ of a possible time-varying dynamic system as $y(t) = T(z,\cdot)[u](t)$. 

\begin{defn}
    A linear operator $T(z,t)$ is stable and proper of 
    \begin{equation}
        |y(t)| = |T(z,\cdot)[u](t)|\leq \beta \sum_{\tau = 0}^{t-1}e^{-\alpha(t-a-\tau)}|u(\tau)| + \gamma|u(\tau)|
    \end{equation}
    for any $u(t)\in \mathbb{R}$, all $t\geq0$, and some constant $\beta>0$, $\alpha>0$, $\gamma=0$. A linear operator $T(z,t)$ is stable and strictly proper if it is stable with $\gamma = 0$.   
\end{defn}

\begin{prop}
    A linear operator $T(z, t)$ is stable and proper if it represents a system described
by the difference equation\begin{equation}
    P(z)[y](t) = Q(z,t)[u](t),
\end{equation}
where $P(z)$ is an $n$th-order constant coefficient polynomial whose zeros are all inside the unit circle of
the complex z-plane, and $Q(z, t)$ is an nth-order polynomial with bounded and possibly time-varying
coefficients. If the order of $Q(z, t)$ is less than n, then $T(z, t)$ is stable and strictly proper.
\end{prop}

\begin{defn}
    A linear operator $T(z, t)$ is nonnegative if $T(z,\cdot)[u](t)\geq0,\,\forall u(t)\geq0, \, \forall t\geq0$. A nonnegative linear operator $T_1(z, t) $ dominates a linear operator $T_2(z, t)$ if
    \begin{equation}
        |T_2(z,\cdot)[u](t)|\leq T_1(z,\cdot)[u](t),\,\forall u(t)\geq0, \, \forall t\geq0. 
    \end{equation}
    A nonnegative linear operator $T(z, t)$ is nondecreasing if
    \begin{equation}
        |T(z,t)[u_1](t)|\leq T(z,t)[u_2](t),\,\forall u_2(t)\geq u_1(t)\geq 0, \, \forall t\geq0. 
    \end{equation}
\end{defn}

\begin{prop}
    For any stable and proper (strictly proper) linear operator $T_2(z, t)$, there exists
a nonnegative, stable and proper (strictly proper) linear operator $T_1(z, t)$ which dominates $T_2(z, t)$.
Such an operator $T_1(z, t)$ can be chosen to be nondecreasing.
\end{prop}

\textbf{Operator-based signal analysis}. For the $j$th element of $u_j(t)$, we have 
\begin{equation}
    u_j(t) = \frac{\Bar{\theta}_j^T(t)x(t) + r_j(t)}{\theta_{2j}(t)}, 
\end{equation}
where $\Bar{\theta}_j(t)$ denotes a part of $\theta_j(t)$ such that $\theta_j(t) = [\Bar{\theta}_j(t)^T,\theta_{2j}(t)]^T$. From Lemma \ref{Lemma: MIMOStabIND} and the parameter projection, both $\frac{\Bar{\theta}_j(t)}{\theta_{2j}(t)}$ and $\frac{1}{\theta_{2j}(t)}$ are bounded. Then, we have the following inequality about the $l^1$ norm of $\omega_j(t) \left[-x^T(t),u_j(t)\right]^T ,\,j = 1,\ldots,m$ defined in \eqref{Def_Omegai}:
\begin{equation}
    \|\omega_j(t)\| \leq \left(1+ \left\|\frac{\Bar{\theta}_j(t)}{\theta_{2j}(t)}\right\|\right) \|x(t)\| + \left\|\frac{r_j(t)}{\theta_{2j}(t)}\right\|, 
\end{equation}
where $\left\|\frac{r_j(t)}{\theta_{2j}(t)}\right\|$ is bounded from the bounedness of $\frac{1}{\theta_{2j}(t)}$.

Thus, for ${\zeta}_{ij}(t) = w_{ij}(z)[{\omega}_j](t)$ in \eqref{ZetaID_MIMO_Diag}, there exists a stable, strictly proper and nonnegative operator $T_{{\zeta}_{ij}}(z)$ such that 
\begin{equation}
    \|{\zeta}_{ij}(t)\| \leq T_{{\zeta}_{ij}}(z)[\|x\|](t) + c_{{\zeta}_{ij}},
    \label{zeta4bound}
\end{equation}
for some constant $c_{{\zeta}_{ij}}> 0 $.

With \eqref{UpperboundNbar} and \eqref{zeta4bound}, there exist a stable, strictly proper and nonnegative operator $T_m(z)$ and a constant $c_m > 0$ such that
\begin{equation}
   a_0 \sum_{i=1}^n \|\zeta_{i}\|\,\vertiii{Q_m}^2 \leq T_m(z)[\|x\|](t) + c_m.
\end{equation}

For the auxiliary signal $\xi_{ij}(t) = {\theta}_j^T(t){\zeta}_{ij}(t)-w_{ij}(z)[{\theta}_j^T {\omega}_j](t) $ in \eqref{XiID_MIMO_Diag}, we set a minimal realization $(A_{ij},b_{ij},c_{ij})$ of $w_{ij}(z) = c_{ij}(zI-A)^{-1}b_{ij}$, with $h_{c,ij}(z) = c_{ij}(zI-A)^{-1}$ and $h_{b,ij}(z) = (zI-A)^{-1}b_{ij}$ both stable and strictly proper. Then, $\xi_{ij}(t)$ is expressed as 
\begin{equation}
\begin{split}
        \xi_{ij}(t) &= {\theta}_j^T(t){\zeta}_{ij}(t) - w_{ij}(z) [{\theta}_j^T {\omega}_j](t)\\
        & = h_{c,ij}(z)[(z-1)[{\theta}_j^T]zh_{b,ij}(z)[{\omega}_j]](t),\\
\end{split}
\end{equation}
where $(z-1)[{\theta}_j](t) ={\theta}_j(t+1) -{\theta}_j(t)\in L^2 $ and $zh_{b,ij}(z)$ is stable and proper.

The above $\xi_{ij}(t)$ can be further expressed as 
\begin{equation}
\begin{split}
        \xi_{ij}(t) &= h_{c,ij}(z)[(z-1)[{\theta}_j^T]h_{b,ij}(z)[z[{\omega}_j]]](t),
        \label{xi_linearoper}
\end{split}
\end{equation}
where $z[{\omega}_j](t) = {\omega_j}(t+1) = \left[-x^T(t+1),u_j(t+1)\right]$ with ${x}(t+1) = A{x}(t)+Bu(t)$. With the parameter projection and $u(t) = \Theta_2^{-1}(t)\Theta_1(t)x(t) +\Theta_2^{-1}(t)r(t)$, $z[{\omega}_j](t)$ in \eqref{xi_linearoper} can be expressed as 
\begin{equation}
\begin{split}
    \left[-\left(A{x} + B\left[\frac{\Bar{\theta}_1}{\theta_{21}},\ldots,\frac{\Bar{\theta}_m}{\theta_{2m}}\right]^Tx + B\rho \right)^T,  \frac{1}{\theta_{2j}^+} \Bar{\theta}_j^{+T} \left(A{x} + B\left[\frac{\Bar{\theta}_1}{\theta_{21}},\ldots,\frac{\Bar{\theta}_m}{\theta_{2m}}\right]^Tx + B\rho \right) + \rho_j^+   \right],
\end{split}
\label{zomega}
\end{equation}
 where $\rho(t) =[\rho_1,\ldots,\rho_m]^T = \left[\frac{r_1(t)}{\theta_{21}(t)},\ldots,\frac{r_m(t)}{\theta_{2m}(t)}\right]^T\in \mathbb{R}^{m}$, $\Bar{\theta}_j^{+}(t) = \Bar{\theta}_j (t+1)$, $\Bar{\theta}_{2j}^{+}(t) = \Bar{\theta}_{2j} (t+1)$, and ${\rho}_j^{+}(t) = {\rho}_j (t+1)$.

From Lemma \ref{Lemma: MIMOStabIND}, we have $\Delta_{{\theta}}(t)={\theta}(t+1)-{\theta}(t)\in L^{2}$ (that is, $\Delta_{\theta_j} (t) = 
\theta_j(t+1) - \theta_j(t) \in L^2$), and the boundedness of $\frac{\Bar{\theta}_j}{\theta_{2j}}$, $\frac{\Bar{\theta}_j^+}{\theta_{2j}^+}$, and $\rho_j(t)$ for $j =1\ldots,m$. Then, for $\xi_{ij}(t)$ in \eqref{xi_linearoper}, there exist stable, strictly proper and nonnegative operators $T_{\xi_{ij} c}(z)$, $T_{\xi_{ij} b}(z,t)$ and a constant $c_{\xi_{c_{ij}}}>0$ such that
\begin{equation}
    |\xi_{ij}(t)|\leq T_{\xi_{ij} c}(z)[{\|\Delta_{{{\theta}_j}}(t)}\|T_{\xi_{ij} b}(z,\cdot)[\|x\|]](t) + c_{\xi_{c_{ij}}}.
\end{equation}

For the corresponding vector $\xi_j = \left[\xi_{1j}(t),\ldots,\xi_{nj}(t)\right]^T$, we have stable, strictly proper and nonnegative operators $T_{\xi_{ij} c}(z)$, $T_{\xi_{j} b}(z,t)$ and a constant $c_{\xi_{c_{j}}}>0$ such that \begin{equation}
    |\xi_{j}(t)|\leq T_{\xi_{j} c}(z)[{\|\Delta_{{{\theta}_j}}(t)}\|T_{\xi_{j} b}(z,\cdot)[\|x\|]](t) + c_{\xi_{c_{j}}}.
\end{equation}
Then, we have the following inequality for $\sum_{i=1}^n \sum_{j=1}^m|\xi_{ij}(t)|$ in \eqref{xl1bound}:
\begin{equation}
    \sum_{i=1}^n \sum_{j=1}^m|\xi_{ij}(t)|\leq \sum_{j=1}^m\left( T_{\xi_{j} c}(z)[{\|\Delta_{{{\theta}_j}}(t)}\|T_{\xi_{j} b}(z,\cdot)[\|x\|]](t) + c_{\xi_{c_{j}}} \right). 
    \label{linearoperate_def_end}
\end{equation}

With \eqref{zeta4bound}-\eqref{linearoperate_def_end}, \eqref{xl1bound} can be rewritten as 
\begin{equation}\begin{split}
    \|{x}(t)\| \leq &\,\|x_m(t)\| + (a_0\vertiii{Q_m}^2 + T_m(z)[\|x\|](t) + c_m)  \left\|\big(\kappa I + Z^T(t)P(t-1)Z(t)\big)^{-\frac{1}{2}}\epsilon(t)\right\|\\& +  \sum_{j=1}^m\left( T_{\xi_{j} c}(z)[{\|\Delta_{{{\theta}_j}}(t)}\|T_{\xi_{j} b}(z,\cdot)[\|x\|]](t) + c_{\xi_{c_{j}}} \right)\\
    \leq & \, \|x_m(t)\| + (a_0\vertiii{Q_m}^2+c_m) \left\|\big(\kappa I + Z^T(t)P(t-1)Z(t)\big)^{-\frac{1}{2}}\epsilon(t)\right\| +\sum_{j=1}^mc_{\xi_{c_j}}\\ &
    +T_m(z)[\|x\|](t) \|\big(\kappa I + Z^T(t)P(t-1)Z(t)\big)^{-\frac{1}{2}}\epsilon(t)\|\\&  +  \sum_{j=1}^m T_{\xi_{j} c}(z)[{\|\Delta_{{{\theta}_j}}(t)}\|T_{\xi_{j} b}(z,\cdot)[\|x\|]](t),
\end{split}
\label{x_bound_almost}
\end{equation}
where $\|x_m(t)\| + (a_0\vertiii{Q_m}^2+c_m) \|\big(\kappa I + Z^T(t)P(t-1)Z(t)\big)^{-\frac{1}{2}}\epsilon(t)\| + \sum_{j=1}^mc_{\xi_{c_j}}$ is bounded. There exists a stable, strictly proper, nonegative and nondecreasing operator $T_0(z)$ such that 
\begin{equation}
    \begin{split}
       & T_m(z)[\|x\|](t) \left\|\big(\kappa I + Z^T(t)P(t-1)Z(t)\big)^{-\frac{1}{2}}\epsilon(t)\right\|  + \sum_{j=1}^m T_{\xi_{j} c}(z)[{\|\Delta_{{{\theta}_j}}(t)}\|T_{\xi_{j} b}(z,\cdot)[\|x\|]](t) \leq  \\
       & \left\|\big(\kappa I + Z^T(t)P(t-1)Z(t)\big)^{-\frac{1}{2}}\epsilon(t)\right\|  T_0(z)[\|x\|](t) + \sum_{j=1}^m T_{\xi_{j} c}(z)[{\|\Delta_{{{\theta}_j}}(t)}\|T_{0}(z,\cdot)[\|x\|]](t)
  .
    \end{split}
    \label{Txib2T0}
\end{equation}

From \eqref{mDefinition}, and \eqref{x_bound_almost}-\eqref{Txib2T0}, it follows that
\begin{equation}
\begin{split}
        T_0(z)[\|x\|](t) &\,\leq      T_0(z)[ \|N^{-\frac{1}{2}}\epsilon \| T_0(z)[ \|x\|]](t)  \\& + T_{0}(z)[ \sum_{j=1}^m T_{\xi_j c}[||\Delta_{{\theta_j}}||T_{0}(z)[\|x\|]]](t) + c_0,
\end{split}
\label{T0Def}
\end{equation}
where $c_0\geq T_0(z)[\|x_m\| + (a_0\vertiii{Q_m}^2+c_m) \|N^{-\frac{1}{2}}\epsilon\| + \sum_{j=1}^mc_{\xi_{c_j}}](t)$
is a constant. For \eqref{T0Def}, there exists a stable and strictly proper operator $T(z)$ such that
\begin{equation}
    \begin{split}
        T_0(z)[\|x\|](t) \leq& \,T(z)[ \|N^{-\frac{1}{2}}\epsilon \| T_0(z)[ \|x\|]](t) +  T(z)[ \sum_{j=1}^m \|\Delta_{\theta_j}\|T_{0}(z
        )[\|x\|]](t) + c_0\\ 
        = &\,T(z)[\left(\|N^{-\frac{1}{2}}\epsilon \| + \sum_{j=1}^m ||\Delta_{{\theta}_j}||\right) T_0(z)[\|x\|]](t) + c_0,
    \end{split}
    \label{xbounded}
\end{equation}
where $\|N^{-\frac{1}{2}}(t)\epsilon(t) \| + \sum_{j=1}^m ||\Delta_{{\theta}_j}(t)|| \in L^2\cap L^{\infty}$. The $L^2$ property of $\|N^{-\frac{1}{2}}(t)\epsilon (t)\| + \sum_{j=1}^m ||\Delta_{{\theta}_j}(t)||$ ensures a small gain for the feedback structure in terms of $T_0(z)[\|x\|](t)$. A small gain theorem can be applied to \eqref{xbounded}, for the boundedness of $T_0(z)[\|x\|](t)$, and so is $x(t)$ from \eqref{x_bound_almost}-\eqref{Txib2T0}. Correspondingly, $u(t)  = \Theta_2^{-1}(t) \left(\Theta^T_1(t)x(t) +\Theta_2^{-1}(t)r(t)\right)$ is also bounded. Thus, all the system signals are bounded. 

Then, $\epsilon(t)\in L^2$ and $\xi_{ij}(t)\in L^2$ from \eqref{xi_linearoper} with $(z-1)[{\theta}](t) = {\theta}(t+1)-{\theta}(t )\in L^2$. Finally, $  x(t)-x_m(t)\in L^2$ so that $\lim_{t\rightarrow\infty} \left(x(t)-x_m(t) \right)= 0. $\hfill$\nabla$

\end{document}